\documentclass{siamltex}
\usepackage{amsfonts, amsmath, amssymb, psfrag, graphicx, bbm, mathrsfs,
enumitem,calc,booktabs}
\usepackage{soul}
\usepackage[colorlinks=true]{hyperref}
\hypersetup{urlcolor=blue, citecolor=blue, linkcolor=blue}

\newtheorem{remark}[theorem]{Remark}

\def\fin{\ifmmode{\Large$\diamond$}\else{\unskip\nobreak\hfil
    \penalty50\hskip1em\null\nobreak\hfil{\Large$\diamond$}
    \parfillskip=0pt\finalhyphendemerits=0\endgraf}\fi}

\def\be#1#2\ee{\begin{equation}\label{eq:#1}#2\end{equation}}
\def\req#1{{\rm(\ref{eq:#1})}}
\def\bdm  {\begin{displaymath}}
  \def\edm  {\end{displaymath}}
\def\bdmal{\begin{displaymath}\begin{aligned}}
    \def\edmal{\end{aligned}\end{displaymath}}

\mathchardef\PhiG="0108
\mathchardef\PsiG="0109
\mathchardef\YG="0107
\mathchardef\Sigma="0106
\mathchardef\Gamma="0100
\mathchardef\Delta="0101
\mathchardef\Lambda="0103
\mathchardef\Omega="010A

\newcommand{\Chat}{{\widehat{C}}}

\renewcommand{\L}{{\mathscr L}}
\newcommand{\N}{{\mathord{\mathbb N}}}
\newcommand{\R}{{\mathord{\mathbb R}}}

\newcommand{\C}{{\mathord{\mathbb C}}}
\newcommand{\E}{{\mathord{\mathbb E}}}

\newcommand{\Normal}{{\cal N}}

\renewcommand{\L}{{\mathscr L}}

\newcommand{\norm}[1]{\|#1\|}

\newcommand{\Real}{{\rm Re\,}}
\newcommand{\Imag}{{\rm Im\,}}

\newcommand{\rmd}{\,\mathrm{d}}
\newcommand{\ds}{\rmd s}

\newcommand{\dt}{\rmd t}

\newcommand{\dtau}{\rmd \tau}

\newcommand{\dW}{\rmd W}

\newcommand{\rmi}{\mathrm{i}}
\newcommand{\eps}{\varepsilon}

\def\rge#1{{\mathcal R}(#1)}
\def\krn#1{{\mathcal N}(#1)}

\def\req#1{{\rm(\ref{eq:#1})}}

\makeatletter
\newcommand{\dupdots}{\mathinner{\mkern1mu\raise\p@
    \vbox{\kern7\p@\hbox{.}}\mkern2mu
    \raise4\p@\hbox{.}\mkern2mu\raise7\p@\hbox{.}\mkern1mu}}
\makeatother
\makeatletter
\newcount\c@MaxMatrixCols \c@MaxMatrixCols=10
\newenvironment{cmatrixc}{%
  \hskip -\arraycolsep\array{*\c@MaxMatrixCols c}%
}{%
  \endarray \hskip -\arraycolsep
}
\makeatother
\newenvironment{cmatrix}{\left[\cmatrixc}{\endmatrix\right]}
 {\endMakeFramed}

\newcommand{\Ftilde}{{\widetilde F}}

\newcommand{\phihat}{{\widehat\varphi}}

\newcommand{\gammahat}{{\widehat\gamma}}

\makeatletter
\renewcommand\@biblabel[1]{#1.}
\makeatother

\title{On data-driven parameterizations of multidimensional generalized Langevin dynamics 
in the presence of a quadratic potential}

\author{M.~Braun\thanks{Institut f\"ur Mathematik, Johannes
    Gutenberg-Universit\"at Mainz, 55099 Mainz, Germany
    ({\tt bmaximi@uni-mainz.de})} 
    \and 
    M.~Hanke\thanks{Institut f\"ur Mathematik, Johannes
    Gutenberg-Universit\"at Mainz, 55099 Mainz, Germany
    ({\tt hanke@math.uni-mainz.de})}
    \and
    N.~Wolf\thanks{Institut f\"ur Chemie, Technische Universit\"at Darmstadt,
    64287 Darmstadt, Germany ({\tt wolf@cpc.tu-darmstadt.de})}
}

\begin{document}
\maketitle

\begin{abstract}
We propose a numerical algorithm to construct a Markov model with an
extended list of variables to parameterize the equation of motion of a
multidimensional coarse-grained physical system in an external potential, when memory
effects are relevant. Our method uses autocorrelation data of the
stationary velocities, but it avoids the inverse problem of finding the
corresponding memory kernel from these data in a first step.
Rather, the data are used to construct a Prony series approximation of 
the autocorrelation function, and the parameters of this Prony series 
provide the corresponding Markov model.
Numerical results for molecular dynamics data 
show a good match for parameterized models with five auxiliary variables 
for a one-dimensional, and twelve auxiliary variables for a two-dimensional system.
\end{abstract}

\begin{keywords} 
coarse-graining, stationary process, external potential, extended Markov model
\end{keywords}

\begin{AMS}
  {\sc 65L09, 65D15, 41A20}
\end{AMS}

%\submitto{\JPCM}
%\bigskip 
%\bigskip
%\noindent
%Version of \today
%\bigskip
%\ioptwocol

\sloppy

%%%%%%%%%%%%%%%%%%%%%%%%%%%%%%%%%%%%%%%%%%%%%%%%%%%%%%%%%%%
\section{Introduction}
\label{Sec:Intro}
%%%%%%%%%%%%%%%%%%%%%%%%%%%%%%%%%%%%%%%%%%%%%%%%%%%%%%%%%%%
The generalized Langevin equation is often used in statistical physics as the
equation of motion of a particle, or a set of particles, when the background
of the system is eliminated from the physical description.
Collecting the spatial coordinates of the centers of mass of the particle(s) of 
interest and their velocities in $d$-dimensional vectors $R$ and $V$, 
respectively, these equations take the form
\begin{equation}
\label{eq:GLE}
\begin{aligned}
   \dot{R}(t) &\,=\, V(t)\,, \\
   M\,\dot{V}(t) 
   &\,=\, -\nabla \PhiG(R) \,-\, \int_0^t\gamma(t-s)V(s)\ds \,+\, F(t)\,,
\end{aligned}
\end{equation}
together with initial conditions at time $t=0$.
Here, $M\in\R^{d\times d}$ is the (symmetric positive definite) mass matrix, 
$\PhiG=\PhiG(R)$ is an external scalar potential,
and the memory kernel $\gamma$ with values in $\R^{d\times d}$ and the 
(stochastic) fluctuating force vector $F\in\R^d$ represent the interactions 
of the particle(s) with the eliminated background of the system; this 
relationship is often associated with a so-called 
fluctuation-dissipation relation, which couples the autocorrelation function 
of the fluctuating force and the memory kernel via
\be{CF}
   C_F(t) \,=\, \E\bigl[F(s+t)F(s)^T\bigr] \,=\, \frac{1}{\beta}\,\gamma(t)\,,
   \qquad s,t\geq 0\,,
\ee
where the constant $\beta>0$ denotes the inverse temperature. We refer to
Zwanzig~\cite{Zwan01} for the physical reasoning behind this model and to
Pavliotis~\cite{Pavl14} for a discussion of \req{GLE}, \req{CF} from a 
mathematical point of view.

In practice, however, an explicit expression for the memory kernel $\gamma$
is usually not available, and the same may be true for the external potential.
A natural way out of the latter problem is to linearize the generalized
Langevin equation and to replace $\PhiG$ by a quadratic potential 
\be{Phi}
   \PhiG(x) \,=\, \frac{1}{2}\,x^T\Omega x\,, \qquad x\in\R^d\,,
\ee
where the symmetric positive definite stiffness matrix $\Omega$ 
is deduced from the sample covariance of the particle positions; 
cf., e.g., Ma, Li, and Liu~\cite{MLL16}, and Remark~\ref{Rem:MLL16} below.

As far as the memory kernel is concerned, we first note that it is quite common
in computational physics to introduce a set of auxiliary variables $Z$ in order 
to replace the non-Markovian generalized Langevin equation~\req{GLE} 
by a Markovian model 
\be{MarkovIntro}
\begin{aligned}
   \dot{X}(t) &\,=\, Y(t)\,, \\[1ex]
   \rmd\begin{cmatrix} Y\\ Z \end{cmatrix}
   &\,=\, \begin{cmatrix} -M^{-1}\nabla\PhiG(X)\\ 0 \end{cmatrix}
          \,+\, \begin{cmatrix} \phantom{i}0 & B^T \\ -C & \Lambda \end{cmatrix}
          \begin{cmatrix} Y\\ Z \end{cmatrix}\!\dt  
          \,+\, \begin{cmatrix} 0 \\ L \end{cmatrix} \dW\,,
\end{aligned}
\ee
with constant matrices $\Lambda$, $B$, $C$, and $L$, 
where $W$ is a $d$-dimensional Brownian motion, and $X$ and $Y$ stand for 
approximations of $R$ and $V$, respectively. Note that this approximation 
is exact, if
\be{gamma-exp}
   \gamma(t) \,=\, M\,B^Te^{t\Lambda}C\,, \qquad t\geq 0\,;
\ee
compare~\cite[Proposition~8.1]{Pavl14} or \req{GLE-approx} below.
Accordingly, the standard approach in the physics literature
(cf., e.g., \cite{GLLB20,KDKSSN19,LBL16}) to find suitable parameters for 
\req{MarkovIntro} is a two-step procedure, which consists in (i) fitting the 
memory kernel $\gamma$ to observation data -- usually correlation data of
$V$ and $\nabla \PhiG(X)$ --, and (ii) subsequently approximating $\gamma$
by an exponential model of the form~\req{gamma-exp}. See also \cite{KWvV24}
for a conceptually different approach.

Note that under the assumption that $\PhiG$ is given by \req{Phi},
the system \req{MarkovIntro} can be rewritten as an Ornstein-Uhlenbeck system
\be{MarkovGoal}
   \rmd\begin{cmatrix} Y\\ Z\\ X \end{cmatrix} 
   \,=\, \begin{cmatrix}
            \,0 & \,B^T & -M^{-1}\Omega \\ -C\, & \Lambda & \phantom{i}0 \\ 
            \,I & 0 & \phantom{i}0
         \end{cmatrix}
         \begin{cmatrix} Y\\ Z\\ X \end{cmatrix}\!\dt
         \,+\, \begin{cmatrix} 0 \\ L \\ 0 \end{cmatrix}\!\dW
\ee
%\be{Markov}
%   \rmd\begin{cmatrix} Y\\ Z\\ X \end{cmatrix}
%   \,=\, A \begin{cmatrix} Y\\ Z \\ X \end{cmatrix}\!\dt \,+\, G \dW
%\ee
of dimension $N>2d$. Our contribution in this paper is an extension of a
method presented earlier in \cite{BBHSH26,Hank24} 
to directly construct a system of the form~\req{MarkovGoal} from a finite
number of equidistant samples of the velocity autocorrelation function only; 
our method does not even need the stiffness matrix $\Omega$ as input, 
in which case the (1,3)-block of the coefficient matrix in \req{MarkovGoal}
will merely be an approximation of $-M^{-1}\Omega$. We mention that
there exist similarities to a method employed by 
Baczewski and Bond~\cite{BaBo13}, and we will briefly elaborate on this in
Remark~\ref{Rem:BaBo13}.

The outline of this paper is as follows. In the next section we discuss the
existence of stationary solutions of the generalized Langevin 
system~\req{GLE}, and we derive some properties of the associated
autocorrelation functions. Subsequently, in Section~\ref{Sec:Markov1},
we present the modifications of the method in \cite{BBHSH26,Hank24} which are
necessary to obtain an Ornstein-Uhlenbeck system in the form~\req{MarkovGoal};
as we will see, by a proper postprocessing of our system we automatically
recover the approximation $X$ of $R$ from the auxiliary variables.
Section~\ref{Sec:Implementation} reviews the numerical ingredients of the 
implementation of our method. Here we aim for a self-contained presentation,
but we omit technical details which have already been
described in previous work. In Subsection~\ref{Subsec:PACF} we also 
briefly comment on the case that position autocorrelation data instead of 
velocity autocorrelation data are given; in principle it is easy to 
modify our algorithm accordingly, but our numerical results 
have been somewhat inferior for this modification. 
%because the high-frequency information content in the velocity data 
%is stronger than in the position data generically.
Finally, Section~\ref{Sec:Numerics} presents numerical
examples using molecular dynamics simulation data for one, 
respectively two, trapped particles as input.

%%%%%%%%%%%%%%%%%%%%%%%%%%%%%%%%%%%%%%%%%%%%%%%%%%%%%%%%%%%
\section{Stationary solutions of the generalized Langevin equation}
\label{Sec:Setting}
%%%%%%%%%%%%%%%%%%%%%%%%%%%%%%%%%%%%%%%%%%%%%%%%%%%%%%%%%%%
In this section we derive sufficient conditions for the existence of a
stationary solution $(R,V)$ of the generalized Langevin equation~\req{GLE}, 
\req{CF}, in the presence of a quadratic potential~\req{Phi}.
Without loss of generality we assume throughout that the mass
matrix in \req{GLE} is of the form 
\be{mass}
   M\,=\,mI
\ee
for some scalar mass value $m>0$;
otherwise we reformulate \req{GLE} as a generalized Langevin equation for
$M^{-1/2}R$ and $M^{-1/2}V$. We will further assume that $\gamma$ 
is extended to all of $\R$ by
\bdm
   \gamma(t) \,=\, \gamma(-t)^T\,, \qquad t<0\,,
\edm
and that this function is absolutely integrable and of strictly positive type,
i.e., its Fourier transform 
\bdm
   \gammahat(\omega) \,=\, \int_{-\infty}^\infty e^{-\rmi\omega t}\gamma(t)\dt
\edm
is positive definite for all $\omega\in\R$. Note that this guarantees
the existence of the fluctuating force term in \req{GLE}, i.e., 
a centered stationary Gaussian process $F$ which satisfies \req{CF} for all
$s,t\in\R$; compare, e.g., Lindgren~\cite{Lind13}.

\begin{proposition}
\label{Prop:FDT2}
Let $\gamma\in L^1(\R)$ be a function of strictly positive type with values
in $\R^{d\times d}$, and let the fluctuating force $F$ be a centered 
stationary Gaussian process satisfying~\req{CF}.
Then the generalized Langevin equation~\req{GLE} with mass matrix~\req{mass}
and initial data
\be{Sigma}
   \begin{cmatrix}
      R(0)\\ V(0)
   \end{cmatrix}
   \,\sim\, {\cal N}(0,\Sigma) \quad \text{with} \quad
   \Sigma \,=\, \frac{1}{\beta m}
                \begin{cmatrix} m\Omega^{-1} & 0 \\ 0 & I \end{cmatrix}
\ee
has a unique solution $(R,V)$,
and this solution is a centered stationary Gaussian process, provided the
initial data are independent of the fluctuating force. 
The autocorrelation function of this solution is given by
\be{C}
   C(t) \,=\, r(t) \Sigma\,, \qquad t\geq 0\,,
\ee
where $r:\R^+_0\to\R^{2d\times 2d}$ solves the Volterra integrodifferential equation
\be{diffresolv}
   r'(t) \,=\, -\begin{cmatrix} 0 & -I \\ \Omega/m & \ 0 \end{cmatrix} r(t)
               \,-\, \int_0^t \begin{cmatrix} 0 & 0 \\ 0 & \gamma(t-s)/m
                              \end{cmatrix} r(s)\ds\,, \qquad
   r(0) \,=\, I\,.
\ee
\end{proposition} 

\begin{proof}
Let
\be{Gamma}
   U \,=\, \begin{cmatrix} R\\ V \end{cmatrix}, \qquad
   \Gamma \,=\, \begin{cmatrix} 0 & 0\\ 0 & \gamma/m \end{cmatrix},
   \qquad  \text{and} \qquad
   D \,=\, \begin{cmatrix} 0 & -I \\ \Omega/m & \phantom{-}0\, \end{cmatrix},
\ee
so that the system~\req{GLE} can be rewritten in the form 
\bdm
   \dot{U} \,=\, -DU(t) \,-\, \int_0^t \Gamma(t-s)U(s)\ds \,+\, \Ftilde(t)\,,
\edm
where
\bdm
   \Ftilde(t) \,=\, \begin{cmatrix} 0 \\ F(t)/m \end{cmatrix}\,.
\edm
According to \req{CF} $\Ftilde$ is a centered stationary Gaussian process with
autocorrelation function
\be{F0}
   C_{\Ftilde}(t) \,=\, \frac{1}{\beta m} \,\Gamma(t)
   \,=\, \Gamma(t)\Sigma\,, \qquad t\geq 0\,,
\ee
with the second identity following readily from the definitions of $\Gamma$ 
and $\Sigma$. 

The assertion of the proposition is now an immediate consequence of the
theory developed in \cite{Hank25}: Since
\bdm
   D\Sigma \,+\, \Sigma D^T \,=\, 0\,,
\edm
the claim follows from \cite[Theorem~2.3]{Hank25} by virtue of \req{F0}.
\end{proof}

\begin{remark}
\label{Rem:MLL16}
\rm
It follows from \req{Sigma} that
\bdm
   \Omega \,=\, \frac{1}{\beta}\,\E\bigl[R(t)R(t)^T\bigr]^{-1}
   \,=\, \frac{1}{\beta}\,C_R(0)^{-1}\,, \qquad
   t\geq 0\,.
\edm
Hence, this provides a means to recover $\Omega$ from the
position autocorrelation function.
\fin
\end{remark}

Rewriting
\be{r}
   r \,=\, \begin{cmatrix} r_{11} & r_{12} \\ r_{21} & r_{22} \end{cmatrix}
\ee
in terms of four $d\times d$ matrix blocks, 
the integro-differential equation~\req{diffresolv} becomes
\be{rij}
\begin{aligned}
   r_{11}'(t) &\,=\, r_{21}(t) \,, \qquad & r_{11}(0)&\,=\,I\,, \\[1ex]
   r_{12}'(t) &\,=\, r_{22}(t) \,, \qquad & r_{12}(0)&\,=\,0\,, \\[1ex]
   m\,r_{21}'(t) &
   \,=\, -\Omega r_{11}(t) \,-\, \int_0^t \gamma(t-s)r_{21}(s)\ds\,, \qquad
                  & r_{21}(0)&\,=\, 0\,, \\[1ex]
   m\,r_{22}'(t) &
   \,=\, -\Omega r_{12}(t) \,-\, \int_0^t \gamma(t-s)r_{22}(s)\ds\,, \qquad
                  & r_{22}(0)&\,=\, I\,.
\end{aligned}
\ee
From this one can derive the following properties of $C_V$.

\begin{theorem}
\label{Thm:CV}
Under the assumptions of Proposition~\ref{Prop:FDT2} 
let $(R,V)$ be the corresponding stationary solution of \req{GLE}.
Then the autocorrelation function $C_V$
%\bdm
%   C_V(t) \,=\, \E\bigl[V(s+t)V(s)^*\bigr] \,, \qquad 
%   t\in\R\,, \ s\geq\max\{0,-t\}\,,
%\edm
of the velocity component of this solution is differentiable at $t=0$ with 
\be{CVprime}
   C_V'(0) \,=\, 0\,,
\ee
and it is absolutely integrable with
\be{CVintegral}
   \int_0^\infty C_V(t) \dt \,=\, 0\,.
\ee
Furthermore, if in addition
\be{gamma-decay}
   \int_0^\infty t\,\bigl|\gamma(t)\bigr|\dt \,<\,\infty\,,
\ee
then
\be{CVMoment1}
   \int_0^\infty t\, C_V(t)\dt \,=\, -\frac{1}{\beta}\,\Omega^{-1}\,.
\ee
\end{theorem}

\begin{proof}
According to \req{C}, \req{diffresolv}, \req{Sigma}, and \req{r} 
the autocorrelation function of the $V$-component of the stationary solution
of \req{GLE} is given by
\bdm
   C_V(t) \,=\, \frac{1}{\beta m}\,r_{22}(t) \qquad \text{for $t\geq 0$\,.}
\edm
By virtue of \req{rij} this implies that $C_V$ is differentiable for $t>0$, 
and the right-sided differential quotient of $C_V$ at $t=0$ is given by
\bdm
   C_V'(0+) \,=\, \lim_{t\to 0\atop t>0} \frac{C_V(t)-C_V(0)}{t}
   \,=\, \frac{1}{\beta m}\,r_{22}'(0) 
   \,=\, -\frac{1}{\beta m^2}\,\Omega\, r_{12}(0) \,=\, 0\,.
\edm
Further, since $C_V(-t)=C_V(t)^T$ for $t>0$, the left-sided difference
quotient $C_V'(0-)$ has the same value. This proves the first claim.

Concerning the integral of $C_V$ we first need to settle that 
$C_V\in L^1(\R^+)$. To this end we investigate the matrices
\be{Azeta}
   H(\zeta) \,=\, \zeta I \,+\, D \,+\, \L\Gamma(\zeta)
\ee
for $\zeta\in\C$ with $\Real\zeta\geq 0$; in \req{Azeta} the matrices
$\Gamma$ and $D$ are taken from \req{Gamma}, and
\bdm
   \L\Gamma(\zeta) \,=\, \int_0^\infty e^{-\zeta t}\Gamma(t)\dt
\edm
denotes the Laplace transform of $\Gamma$, which is well defined for
$\Real\zeta\geq 0$, because $\Gamma\in L^1(\R^+)$. Let $x,y\in\C^d$ 
be such that
\be{Ax=0}
   H(\zeta_0)\begin{cmatrix} x \\ y \end{cmatrix}
   \,=\, \begin{cmatrix} 
            \zeta_0 I & -I \\ \Omega/m & \zeta_0 I + \L\gamma(\zeta_0)/m
         \end{cmatrix}
         \begin{cmatrix} x \\ y \end{cmatrix}
   \,=\, \begin{cmatrix} 0 \\ 0 \end{cmatrix}
\ee
for any given $\zeta_0$ with $\Real\zeta_0\geq 0$. 
Then the scalar complex-valued function
\bdm
   f(\zeta) \,:=\, \frac{1}{\beta}
                  \begin{cmatrix} x \\ y \end{cmatrix}^*
                  \Sigma^{-1}H(\zeta)
                  \begin{cmatrix} x \\ y \end{cmatrix}
   \,=\, \begin{cmatrix} x \\ y \end{cmatrix}^*
         \begin{cmatrix} 
            \zeta \Omega & -\Omega \\ \Omega & m\zeta I + \L\gamma(\zeta)
         \end{cmatrix}
         \begin{cmatrix} x \\ y \end{cmatrix},
\edm
which is analytic in the open right complex half plane, vanishes for 
$\zeta=\zeta_0$; furthermore, on the imaginary axis, 
i.e., for $\zeta=\rmi\omega$ with $\omega\in\R$, its real part is given by
\bdmal
   \Real f(\rmi\omega) 
   &\,=\, \frac{1}{2\beta}
          \begin{cmatrix} x \\ y \end{cmatrix}^*
          \Bigl(\Sigma^{-1}H(\rmi\omega) 
                \,+\, \bigl(\Sigma^{-1}H(\rmi\omega)\bigr)^*\Bigr)
          \begin{cmatrix} x \\ y \end{cmatrix} \\[1ex]
   &\,=\, \frac{1}{2} \begin{cmatrix} x \\ y \end{cmatrix}^*
          \begin{cmatrix} 
             0 & 0 \\ 
             0 & \L\gamma(\rmi\omega)\,+\,\bigl(\L\gamma(\rmi\omega)\bigr)^*
          \end{cmatrix}
          \begin{cmatrix} x \\ y \end{cmatrix}
    \,=\, \frac{1}{2}\, y^* \gammahat(\omega) y\,.
\edmal
Therefore, as $\gamma$ is assumed to be strictly positive real, the
real part of $f$ is either zero on the entire imaginary axis when $y=0$, 
or it is positive on the entire imaginary axis when $y\neq 0$. 
Since $\Real f$ is a harmonic function on the right half plane, it follows 
from the maximum principle that in the latter case $\Real f$ is positive in
the entire right half plane; accordingly, $f(\zeta_0)=0$ implies that $y=0$.

Inserting this into \req{Ax=0} we see that
\bdm
   \begin{cmatrix} 0 \\ 0 \end{cmatrix}
   \,=\, H(\zeta_0) \begin{cmatrix} x \\ y \end{cmatrix}
   \,=\, \begin{cmatrix} \zeta_0 x \\ \Omega x/m \end{cmatrix},
\edm
which shows that $x=0$, too, because $\Omega$ is nonsingular. In other words, 
the null space of $H(\zeta_0)$ is trivial, and we therefore have established
that $H(\zeta)$ is nonsingular for every $\zeta$ in the closed right half plane.
Now it is known
(compare Gripenberg, London, and Staffans~\cite[Theorem~3.3.5]{GLS90})
that this is equivalent to the fact that the (unique) solution $r$ of 
\req{diffresolv} belongs to $L^1(\R^+)$, and so does $r'$.
From this it readily follows that $r(t)\to 0$ as $t\to\infty$.
We thus have shown that $C_V\in L^1(\R^+)$, and \req{rij} implies that
\be{CVint-tmp}
   \int_0^\infty C_V(t)\,dt \,=\, \frac{1}{\beta m} \int_0^\infty r_{22}(t)\dt
   \,=\, \frac{1}{\beta m}\,r_{12}(t)\Big|_{t=0}^\infty \,=\, 0\,.
\ee

If $\gamma$ also satisfies the decay condition~\req{gamma-decay} then
\cite[Theorem~4.4.13]{GLS90} implies that the same decay condition
is true for $r$ and $r'$, and hence, $tr(t)\to 0$ as $t\to\infty$.
We therefore obtain by partial integration that
\bdmal
   \int_0^\infty t\,C_V(t)\dt 
   &\,=\, \frac{1}{\beta m}\int_0^\infty t\, r_{22}(t)\dt
    \,=\, \frac{1}{\beta m}\, t\,r_{12}(t)\Big|_{t=0}^\infty 
          \,-\, \frac{1}{\beta m} \int_0^\infty r_{12}(t)\dt\\[1ex]
   &\,=\, \,-\, \frac{1}{\beta m} \int_0^\infty r_{12}(t)\dt\,.
\edmal
It further follows from \req{rij} that
\bdm
   r_{12}(t) \,=\, -\Omega^{-1}\left(m\,r_{22}'(t) 
                   \,+\, \int_0^t \gamma(t-s)r_{22}(s)\ds\right),
\edm
and inserting this into the previous identity yields
\bdmal 
   \int_0^\infty t\,C_V(t)\dt
   &\,=\, \frac{1}{\beta m}\,\Omega^{-1}
          \left(\int_0^\infty m\,r_{22}'(t)\dt \,+\,
                \int_0^\infty\int_0^t \gamma(t-s)r_{22}(s)\ds\dt
          \right)\\[1ex]
   &\,=\, \frac{1}{\beta m}\,\Omega^{-1}
          \left( m\,r_{22}(t)\Big|_{t=0}^\infty \,+\,
                \int_0^\infty \gamma(\tau)\dtau \int_0^\infty r_{22}(s)\ds
          \right)\,.
\edmal
According to \req{CVint-tmp} the integral over $r_{22}$ vanishes, and hence,
under the given assumptions, we arrive at
\bdm   
   \int_0^\infty t\,C_V(t)\dt \,=\, -\frac{1}{\beta}\,\Omega^{-1} r_{22}(0)
   \,=\, -\frac{1}{\beta}\, \Omega^{-1}
\edm
by virtue of \req{rij} again.
\end{proof}

\begin{remark}
\label{Rem:selfdiffusion}
\rm
In statistical physics the integral over $C_V$ is known as the 
self-diffusion coefficient of the macroparticle, cf., e.g., \cite{Zwan01}. 
This coefficient being equal to zero according to \req{CVintegral}
reflects the fact that the position of the macroparticle is a centered 
stationary process.
\fin
\end{remark}

%%%%%%%%%%%%%%%%%%%%%%%%%%%%%%%%%%%%%%%%%%%%%%%%%%%%%%%%%%%
\section{A Markov approximation}
\label{Sec:Markov1}
%%%%%%%%%%%%%%%%%%%%%%%%%%%%%%%%%%%%%%%%%%%%%%%%%%%%%%%%%%%
Theorem~\ref{Thm:CV} implies that the autocorrelation
function of the velocity component of the stationary solution of \req{GLE}
carries the entire information about the quadratic potential.
Let us accordingly assume that finitely many equidistant samples
$C_V(\nu\tau)$, $\nu=0,\dots,n$, of the autocorrelation function 
are given for some $\tau>0$ and $n\in\N$;
%of the velocity of the macroparticle are given; 
as stated in the introduction our goal is to set up an 
Ornstein-Uhlenbeck equation
\be{Markov}
   \rmd\begin{cmatrix} Y\\ Z\\ X \end{cmatrix} 
   \,=\, \begin{cmatrix}
            \,0 & \,B^T & -\Omega_0/m \\ -C\, & \Lambda & \phantom{i}0 \\ 
            \,I & 0 & \phantom{i}0
         \end{cmatrix}
         \begin{cmatrix} Y\\ Z\\ X \end{cmatrix}\!\dt
         \,+\, \begin{cmatrix} 0 \\ L \\ 0 \end{cmatrix}\!\dW
\ee
of dimension $N>2d$ with auxiliary variables $Z\in\R^{N-2d}$ and 
$\Omega_0\approx\Omega$,
such that the components $Y\in\R^d$ and $X\in\R^d$ of its stationary solution
satisfy $Y\approx V$ and $X\approx R$ in the sense that their correlations
are in good agreement. 
%In \req{Markov} $A$ and $G$ are real $N\times N$
%and $N\times d$ matrices, respectively, 
%and $W$ is a $d$-dimensional Brownian motion.

To this end we employ the method developed in \cite{Hank24} 
for scalar stochastic processes 
and extended in \cite{BBHSH26} to the case of vector-valued processes;
this method will be briefly reviewed in Section~\ref{Sec:Implementation}.
It is designed to provide an Ornstein-Uhlenbeck system
\be{Markov0}
   \rmd\begin{cmatrix} Y \\ Z_0 \end{cmatrix} 
   \,=\, \begin{cmatrix} \,0 & B_0^T \\ -C_0 & A_0\, \end{cmatrix}
         \begin{cmatrix} Y \\ Z_0 \end{cmatrix}\!\dt
         \,+\, \begin{cmatrix} 0 \\ L_0 \end{cmatrix}\!\dW
\ee
of dimension $N$ with auxiliary variables $Z_0\in\R^{N-d}$. 
The real drift matrix of \req{Markov0}, denoted by $A$ in the sequel, 
is chosen to be stable, i.e., its eigenvalues belong to
\bdm 
   \C^- \,=\, \{\,\lambda\in\C\,:\, \Real\lambda<0\,\}\,,
\edm
and the matrix $L_0\in\R^{(N-d)\times d}$ is constructed in such a way that the 
(unique) stationary solution of \req{Markov0} satisfies
\be{variances}
   C_Y(0) \,=\, \E\bigl[Y(0)Y(0)^T\bigr] \,=\, C_V(0) 
          \,=\, \frac{1}{\beta m}\,I \quad \text{and} \quad
   \E\bigl[Y(0)Z_0(0)^T\bigr] \,=\, 0\,.
\ee

\begin{remark}
\label{Rem:BaBo13}
\rm
We recall that the system~\req{Markov0} can be recast as a generalized
Langevin equation
\bdm 
   Y'(t) \,=\, -\int_0^t \gamma_{\rm eff}(t-s)Y(s)\ds \,+\, F_{\rm eff}(t)
\edm
for $Y$ with an effective memory kernel $\gamma_{\rm eff}$ and associated
fluctuating force $F_{\rm eff}$, which comes without any external potential. 
The idea of using an effective equation of this sort for $Y\approx V$, 
when $\PhiG$ in \req{GLE} is given by \req{Phi}, has been the 
starting point of the construction of Baczewski and Bond in \cite{BaBo13}; 
they employ the Laplace transform to connect the effective 
memory kernel to the given velocity autocorrelation function $C_V$. 
The advantage of our approach is that it avoids the (ill-posed)
inverse Laplace transform 
and determines suitable parameters for \req{Markov0} directly from the data.
\fin
\end{remark}

The $d\times d$ zero block in the upper left corner of the drift matrix 
$A$ in \req{Markov0} is a consequence of the side condition
\be{CYprime}
   C_Y'(0) \,=\, 0\,,
\ee
which is imposed in \cite{BBHSH26,Hank24} to mimic \req{CVprime}. 
In contrast to the method as it is presented in \cite{BBHSH26,Hank24},
here we further impose the additional constraint
\be{CYintegral}
   \int_0^\infty C_Y(t)\dt \,=\, 0
\ee
because of \req{CVintegral}; see Section~\ref{Sec:Implementation} for details. 
From \req{Markov0} and \req{variances} it follows that
\be{CY}
   C_Y(t) \,=\, \frac{1}{\beta m}\,E^T e^{tA} E \qquad \text{with} \qquad
   E \,=\, \begin{cmatrix} I \\ 0 \end{cmatrix}
\ee
for $t\geq 0$, cf., e.g., \cite[Section~3.7]{Pavl14}, and hence, 
\req{CYprime} and \req{CYintegral} are equivalent to
\be{EM1E}
   E^TAE \,=\, 0 \qquad \text{and} \qquad E^TA^{-1}E \,=\, 0\,,
\ee
respectively. 
%We will also enforce that $C_Y(0)=C_V(0)$ exactly, i.e.,
%\be{EE}
%   C_Y(0) \,=\, \frac{1}{\beta m}\,E^TE 
%   \,=\, \frac{1}{\beta m}\,I\,,
%\ee
%compare~\req{variances}.

Since $C_Y(t)=C_Y(-t)^T$ for $t<0$, \req{CY} implies that 
the Fourier transform of $C_Y$ is given by
\be{CYhat-kappa}
\begin{aligned}
   \Chat_Y(\omega) 
   &\,=\, \frac{1}{\beta m}\int_{-\infty}^0 e^{-\rmi\omega t}E^Te^{-tA^T}E\dt
          \,+\, \frac{1}{\beta m}\int_0^\infty e^{-\rmi\omega t}E^T e^{tA}E\dt
          \\[1ex]
   &\,=\, \frac{1}{\beta m}
          \bigl(\kappa(\rmi\omega)^* + \kappa(\rmi\omega)\bigr)\,,
\end{aligned}
\ee
where 
\be{kappa}
   \kappa(\zeta) \,=\, E^T(\zeta I-A)^{-1}E
\ee
is the so-called transfer function of the system~\req{Markov0}.
This is a real rational function, i.e., 
$\kappa(\overline{z})=\overline{\kappa(z)}$ for all $z\in\C$,
which is analytic in a neighborhood of the
closed right half-plane because $A$ is a stable matrix.
Further, since every autocorrelation function is a 
function of positive type, cf., e.g., \cite[Lemma~1.2]{Pavl14}, 
it follows from \req{CYhat-kappa} that
the Hermitian part of $\kappa(\zeta)$ is positive semidefinite for every 
$\zeta$ on the imaginary axis. In system theory a function $\kappa$ with
these properties is said to be positive real;
compare Anderson and Vongpanitlerd~\cite{AV73}.

\begin{theorem}
\label{Thm:A0}
Let the drift matrix $A$ of \req{Markov0} be stable, and let \req{Markov0}
have a stationary solution which satisfies \req{variances}.
Assume further that $E^TA^{-1}E=0$ with $E$ as in \req{CY}.
Then the matrix block $A_0$ in \req{Markov0} has a $d$-dimensional 
null space $\krn{A_0}$, and this is a complementary subspace of the
range space $\rge{A_0}\subset\R^{N-d}$.
\end{theorem}

\begin{proof}
We start by noting that $B_0$ and $C_0$ must have full rank, because $A$
is stable, and hence nonsingular. For the same reason
the null spaces of $A_0$ and $B_0^T$ must not have a
nontrivial vector $z$ in common, for otherwise $[0;z]$ 
is a nontrivial element of the null space of $A$.
It follows that the dimension of the null space of $A_0$ can be at most $d$,
for the null space of $B_0^T$ already has dimension $N-d$.

Consider now any vector
\bdm
   w \,\in\, \rge{A_0} \,\cap\, \rge{C_0}\,.
\edm
Then there exist $z\in\R^{N-d}$ and $y\in\R^d$ with
\bdm
   w \,=\, A_0z \,=\, C_0y\,,
\edm
and hence,
\be{M-LGS}
   A \begin{cmatrix} y\\ z \end{cmatrix} 
   \,=\, \begin{cmatrix} \,0 & B_0^T \\ -C_0 & A_0\, \end{cmatrix}
         \begin{cmatrix} y\\ z \end{cmatrix}
   \,=\, \begin{cmatrix} B_0^Tz\\ 0 \end{cmatrix}.
\ee
Due to the assumption that $E^TA^{-1}E=0$, it follows from \req{M-LGS} that
\bdm 
   0 \,=\, E^TA^{-1}E\,B_0^Tz 
     \,=\, E^TA^{-1}\begin{cmatrix} B_0^Tz\\ 0 \end{cmatrix}
     \,=\, E^T\begin{cmatrix} y\\ z \end{cmatrix} 
     \,=\, y\,,
\edm
and this in turn implies that $w=C_0y=0$. Accordingly, the range spaces of
$A_0$ and $C_0$ are complementary spaces, and hence, 
$\rge{A_0} \oplus \rge{C_0}$ has dimension $(N-d-\dim\krn{A_0})+d$.
Since this cannot be larger than $N-d$, we necessarily have 
$\dim\krn{A_0}\geq d$, and therefore the dimension of the null space of 
$A_0$ is exactly equal to $d$.

For the final assertion let us assume that there is some nontrivial
element $z\in\rge{A_0}\cap\krn{A_0}$. Then there exists $x\in\R^{N-d}$, 
such that
\be{Jordanblock}
   A_0x \,=\, z \qquad \text{and} \qquad A_0z \,=\, 0\,,
\ee
and since $B_0^T$ is a bijection from $\krn{A_0}$ to $\R^d$ there further
exists $x'\in\krn{A_0}$ with
\bdm
   B_0^T x' \,=\, B_0^T x\,.
\edm
Accordingly,
\bdm
   A \begin{cmatrix} 0\\ x-x' \end{cmatrix} 
   \,=\, \begin{cmatrix} \,0 & B_0^T \\ -C_0 & A_0\, \end{cmatrix}
         \begin{cmatrix} 0\\ x-x' \end{cmatrix}
   \,=\, \begin{cmatrix} B_0^T(x-x')\\ A_0x-A_0x' \end{cmatrix}
   \,=\, \begin{cmatrix} 0\\ z \end{cmatrix},
\edm
and hence,
\bdm
   A^2 \begin{cmatrix} 0\\ x-x' \end{cmatrix} 
   \,=\, A \begin{cmatrix} 0\\ z \end{cmatrix}
   \,=\, \begin{cmatrix} \,0 & B_0^T \\ -C_0 & A_0\, \end{cmatrix}
         \begin{cmatrix} 0\\ z \end{cmatrix}
   \,=\, \begin{cmatrix} B_0^Tz\\ A_0z\end{cmatrix}
   \,=\, \begin{cmatrix} B_0^Tz\\ 0 \end{cmatrix}.
\edm
This shows that for $y=B_0^Tz\in\R^d$ there holds
\be{yEM2Ey}
   y^TE^TA^{-2}Ey 
   \,=\, y^TE^T A^{-2} \begin{cmatrix} B_0^Tz\\ 0 \end{cmatrix} 
   \,=\, y^TE^T\begin{cmatrix} 0\\ x-x' \end{cmatrix} \,=\, 0\,.
\ee
Note that $y\neq 0$ because $z\neq 0$ and the intersection of the null spaces
of $B_0^T$ and $A_0$ is trivial.

As before, let $C_Y$ be the autocorrelation function of the $Y$-component
of the stationary solution of \req{Markov0}, and $\kappa$ be given by
\req{kappa}. Since $\kappa$ is positive real, the scalar rational function
\be{f}
   f(\zeta) \,:=\, y^*\kappa(\zeta)y
\ee
with $y$ of \req{yEM2Ey} is analytic
in a neighborhood of the closed right-half complex plane with
%moreover, \req{CYhat-kappa} implies that
%\bdm
%   2\,\Real f(\rmi\omega) 
%   \,=\, y^T\bigl(\kappa(\rmi\omega)\,+\,\kappa(\rmi\omega)^*\bigr)y
%   \,=\, \beta m\,y^T\Chat_Y(\omega)y \,\geq\,0 \qquad
%   \text{for every $\omega\in\R$}\,,
%\edm
%because $C_Y$ is a function of positive type.
\bdm
   \Real f(\rmi\omega) \,\geq\, 0 \qquad
   \text{for every $\omega\in\R$}\,.
\edm
%because $\kappa$ is assumed to be positive real.
By the maximum principle for harmonic functions, 
$\Real f$ must therefore be positive 
in the open right-half plane, for if it has a zero at some $\zeta_0$ with 
$\Real\zeta_0>0$, then $\Real f$ is identically zero, in contradiction to
\bdm
   \lim_{\zeta\to\infty} \zeta f(\zeta) \,=\, \norm{y}_2^2\,>\, 0\,.
\edm
%and hence,
%\bdm
%   \Chat_Y(\omega)y \,=\, 0 \qquad \text{for all $\omega\in\R$}\,;
%\edm
%but this would imply that $C_Y(t)y=0$ for all $t\geq 0$, in contradiction
%to \req{variances}.

On the other hand, \req{kappa}, \req{EM1E}, and \req{yEM2Ey} imply that
\bdm
   f(0) \,=\, -y^*E^TA^{-1}Ey \,=\, 0 \qquad \text{and} \qquad
   f'(0) \,=\, -y^*E^TA^{-2}Ey \,=\, 0\,,
\edm
i.e., $f$ has a multiple zero at $\zeta=0$. Due to the local behavior of
analytic functions, cf., e.g., Ahlfors~\cite[Section~4.3]{Ahlf79}, 
this contradicts the fact that its real part is positive in the
intersection of any small neighborhood of $\zeta=0$ and the 
open right half-plane. This proves that no nontrivial vector $z$ 
satisfying~\req{Jordanblock} can exist, and therefore $\rge{A_0}$ and 
$\krn{A_0}$ are complementary subspaces.
\end{proof}

Consider the coordinate transformation which turns $A_0$ into its 
real Jordan canonical form. We split accordingly the solution component 
$Z_0$ of \req{Markov0} into new variables $Z\in\R^{N-2d}$ corresponding to the 
nonzero eigenvalues of $A_0$ (its range space) and the remaining 
eigencomponents $X_0\in\R^d$ from the null space of $A_0$; 
denoting by $\Lambda\in\R^{(N-2d)\times(N-2d)}$ the Jordan block associated 
with the nonzero eigenvalues of $A_0$ we can thus rewrite \req{Markov0} as
\bdm
   \rmd\begin{cmatrix} Y\,\\ Z\,\\ X_0 \end{cmatrix}
   \,=\, \begin{cmatrix}
            \,0 & \,B^T & B_1^T \\ -C\phantom{_1} & \Lambda & 0 \\ -C_1 & 0 & 0
         \end{cmatrix}
         \begin{cmatrix} Y\,\\ Z\,\\ X_0 \end{cmatrix}\!\dt
         \,+\, \begin{cmatrix} 0 \\ L \\ L_1 \end{cmatrix}\!\dW\,.
\edm
Since the drift matrix of this equation is similar to the matrix $A$ 
in \req{Markov0}, it is also nonsingular, and hence, $C_1\in\R^{d\times d}$ is 
nonsingular. Accordingly, introducing $X=-C_1^{-1}X_0$ we obtain the equivalent
system
\be{Markov1}
   \rmd\begin{cmatrix} Y\\ Z\\ X \end{cmatrix} 
   \,=\, \begin{cmatrix}
            \,0 & \,B^T & -\Omega_0/m \\ -C\, & \Lambda & \phantom{i}0 \\ 
            \,I & 0 & \phantom{i}0
         \end{cmatrix}
         \begin{cmatrix} Y\\ Z\\ X \end{cmatrix}\!\dt
         \,+\, \begin{cmatrix} 0 \\ L \\ L_2 \end{cmatrix}\!\dW
\ee
with
\[
   \Omega_0 \,=\, m\,B_1^TC_1 \qquad \text{and} \qquad L_2\,=\, -C_1^{-1}L_1\,.
\]
This is the Langevin equation~\req{Markov} that we have been heading for.
%and without loss of generality we can therefore assume in the sequel that the
%drift matrix $A$ of \req{Markov0} has the block form given in \req{Markov1}.

\begin{proposition}
\label{Prop:Markov1}
Under the assumptions of Theorem~\ref{Thm:A0}
%Let the drift matrix $M$ of \req{Markov0} be stable, and let \req{Markov0}
%have a stationary solution which satisfies \req{variances}.
%Assume further that $E^TM^{-1}E=0$ with $E$ as in \req{CY}.
the matrix block $\Omega_0$ in \req{Markov1} is 
symmetric positive definite and $L_2$ is equal to zero. 
Furthermore, \req{Markov1} admits a stationary solution, 
the covariance matrix of which is given by
\[
   \Sigma_0 \,=\,
   \frac{1}{\beta m}\,
   \begin{cmatrix}
      I & 0 & 0 \\ 0 & \Sigma_{22} & 0 \\
      0 & 0 & m\Omega_0^{-1}
   \end{cmatrix},
\]
where the symmetric positive semidefinite matrix block 
$\Sigma_{22}\in\R^{(N-2d)\times(N-2d)}$ satisfies
\be{Lure-small}
   \Lambda\Sigma_{22} \,+\, \Sigma_{22}\Lambda^T \,=\, -\beta m\,LL^T
   \qquad \text{and} \qquad
   \Sigma_{22} B \,=\, C
\ee
with $B$ and $C$ of \req{Markov1}.
%Moreover, all three solution components $Y(t)$, $Z(t)$, and $X(t)$ are 
%uncorrelated, and
%\be{Prop:Markov1}
%   \E\bigl[X(t)X(t)^T\bigr] \,=\, \frac{1}{\beta}\,\Omega_0^{-1}\,.
%\ee
\end{proposition}

\begin{proof}
We have already seen that under the given assumptions \req{Markov1} admits
a stationary solution, whose components $Z(t)$ and $X(t)$ are 
uncorrelated with $Y(t)$ because of \req{variances}. Accordingly, the
covariance matrix $\Sigma_0$ of this stationary solution 
of \req{Markov1} is given by
\bdm
   \Sigma_0 \,=\,
   \frac{1}{\beta m}\,
   \begin{cmatrix}
      I & 0 & 0 \\ 0 & \Sigma_{22} & \Sigma_{23} \\
      0 & \Sigma_{23}^T & \Sigma_{33}
   \end{cmatrix}
\edm
for appropriate matrix blocks $\Sigma_{22}$, $\Sigma_{23}$, 
$\Sigma_{33}$,
and the associated Lyapunov equation for $\Sigma_0$ takes the form 
\begin{align*}
   &\begin{cmatrix}
      \,0 & \,B^T & -\Omega_0/m \\ -C\, & \Lambda & \phantom{i}0 \\ 
      \,I & 0 & \phantom{i}0
   \end{cmatrix}
   \begin{cmatrix}
      I & 0 & 0 \\ 0 & \Sigma_{22} & \Sigma_{23} \\
      0 & \Sigma_{23}^T & \Sigma_{33}
   \end{cmatrix}
   \,+\,
   \begin{cmatrix}
      I & 0 & 0 \\ 0 & \Sigma_{22} & \Sigma_{23} \\
      0 & \Sigma_{23}^T & \Sigma_{33}
   \end{cmatrix}
   \begin{cmatrix}
      \,0 & \!\!-C^T & I\, \\ B & \!\!\phantom{-}\Lambda^T\, & 0\, \\ 
      -\Omega_0^T/m & \!\!0 & 0\,
   \end{cmatrix}\\
   &\quad =\, -\beta m
              \begin{cmatrix} 
                 0 & 0 & 0 \\ 0 & LL^T & LL_2^T\\ 0 & L_2L^T & L_2L_2^T
              \end{cmatrix}.
\end{align*}
Evaluating the (3,3) block term on the left-hand side of this equation
it readily follows that $L_2=0$. A subsequent evaluation of the (2,3)
block term yields
\bdm
   \Lambda\Sigma_{23} \,=\, -\beta m\, LL_2^T \,=\, 0\,,
\edm
and since $\Lambda$ is nonsingular, it follows that $\Sigma_{23}=0$.

To determine $\Sigma_{33}$ we compute the (3,1) block identity in the 
Lyapunov equation:
\bdm
   I \,+\, \Sigma_{23}^T B \,-\, \Sigma_{33}\Omega_0^T/m \,=\, 0\,.
\edm
Since we already know that $\Sigma_{23}=0$ we conclude that 
\bdm
   \Sigma_{33} \Omega_0^T \,=\, m\,I\,.
\edm
This implies that $\Sigma_{33}$ and $\Omega_0$ have full rank, and that
\bdm 
   \Sigma_{33} \,=\, m\,\Omega_0^{-T}\,.
\edm
Clearly, $\Omega_0$ must therefore be symmetric and positive definite.
Finally, an evaluation of the (2,2) and (2,1) block terms gives the remaining 
two equations~\req{Lure-small} for $\Sigma_{22}$.
\end{proof}

The second block equation of \req{Markov1} yields the representation
\bdm 
   Z(t) \,=\, e^{t\Lambda} Z(0) \,-\, \int_0^t e^{(t-s)\Lambda}CY(s)\ds
              \,+\, \int_0^t e^{(t-s)\Lambda} L\dW(s)\,, \qquad t\geq 0\,,
\edm
of the auxiliary $Z$-variables of \req{Markov1}.
Inserting this expression into the first block equation of \req{Markov1}
and making use of Proposition~\ref{Prop:Markov1} we conclude that $X$ and $Y$
satisfy the system of stochastic differential equations
\be{GLE-approx}
\begin{aligned}
   \dot{X}(t) &\,=\, Y(t)\,, \\
   m\,\dot{Y}(t) &\,=\, -\Omega_0X(t) 
                  \,-\, \int_0^t m\,B^Te^{(t-s)\Lambda}CY(s)\ds 
                     \,+\, F_0(t)\,,
\end{aligned}
\ee
with
\[
   F_0(t) \,=\, m\,B^Te^{t\Lambda}Z(0) 
                 \,+\, \int_0^t m\,B^Te^{(t-s)\Lambda} L\dW(s)\,.
\]
This system has the same structure as the original model~\req{GLE} that we
have started with, and the memory kernel $\gamma_0$ associated with
\req{GLE-approx} has the form
\bdm
   \gamma_0(t) \,=\, m\,B^Te^{t\Lambda}C\,, \qquad t\geq 0\,,
\edm
which is absolutely integrable whenever $\Lambda$ is a stable matrix;
see Remark~\ref{Rem:Lambda} below. It can further be shown 
(compare, e.g., the computation in \cite[Section~8.2]{Pavl14})
that in this case the fluctuating force $F_0$ is a 
centered stationary Gaussian process with autocorrelation function
\bdm
   C_{F_0}(t) \,=\, \E\bigl[F_0(s+t)F_0(s)^T\bigr] 
   \,=\, \frac{1}{\beta}\,\gamma_0(t)\,, \qquad t\geq 0\,,
\edm
provided that
\bdm
   Z(0) \,\sim\, \Normal\Bigl(0,\frac{1}{\beta m}\,\Sigma_{22}\Bigr)\,,
\edm
independent of the Brownian motion $W$. This should be compared to \req{CF}.

\begin{remark}
\label{Rem:Lambda}
\rm
Recall that the autocorrelation function $C_Y$ of the $Y$-component of the
stationary solution of \req{Markov1} is a function of positive type, 
but beware of the fact that it fails to be of strictly positive type 
under the assumptions of Theorem~\ref{Thm:A0}
%cf., e.g., \cite[Lemma~1.2]{Pavl14}, 
%i.e., $\Chat_Y(\omega)$ is positive semidefinite for every $\omega\in\R$,
%and $\Chat_Y(0)=0$ by virtue of \req{CYintegral}.
%and 
because $\widehat{C}_Y(0)=0$ by virtue of \req{CYintegral}.
%is twice the symmetric part of the 
%(real) integral~\req{CYintegral}, it follows that $\Chat_Y(0) = 0$, 
%showing that $C_Y$ is not of strictly positive type.

In this situation a sufficient condition for $\Lambda$ to be 
stable is that $\Chat_Y(\omega)$
%which is a function of positive type, 
%cf., e.g., Pavliotis~\cite[Lemma~1.2]{Pavl14},
is positive definite for every $\omega\neq 0$. 
A simple argument for this statement is as follows. Assume that $\Lambda$ has
a nonzero eigenvalue $\lambda$ on the imaginary axis or in the open right 
half-plane. Let $z\in\C^{N-d}$ be an associated eigenvector of $A_0$,
and denote the coefficient matrix of \req{Markov0} once again by $A$. 
Then $B_0^Tz$ must be different from zero, for otherwise 
\bdm 
   A \begin{cmatrix} 0 \\ z \end{cmatrix} 
   \,=\, \begin{cmatrix} \,0 & B_0^T \\ -C_0 & A_0\, \end{cmatrix}
         \begin{cmatrix} 0 \\ z \end{cmatrix} 
   \,=\, \begin{cmatrix} B_0^Tz \\ A_0z \end{cmatrix}
   \,=\, \begin{cmatrix} 0 \\ \lambda z \end{cmatrix},
\edm
and hence, $\lambda$ also belongs to the spectrum of $A$; since $A$ is stable,
this is not possible. Accordingly, 
\bdm
   (\lambda I-A) \begin{cmatrix} 0 \\ z \end{cmatrix}
   \,=\, \begin{cmatrix} \lambda I & -B_0^T \\ C_0 & \lambda I-A_0 \end{cmatrix}
         \begin{cmatrix} 0 \\ z \end{cmatrix}
   \,=\, \begin{cmatrix} -B_0^Tz \\ 0 \end{cmatrix},
\edm
and hence, for $y=-B_0^Tz\neq 0$ it follows that
\bdm
   y^*E^T(\lambda I-A)^{-1}E \, y 
   \,=\, y^*E^T (\lambda I-A)^{-1} \begin{cmatrix} -B_0^Tz \\ 0 \end{cmatrix}
   \,=\, y^*E^T \begin{cmatrix} 0 \\ z \end{cmatrix}
   \,=\, 0\,.
\edm
This means that the rational function $f$ defined in \req{f} vanishes at 
$\zeta=\lambda$. As discussed in the proof of Theorem~\ref{Thm:A0}, however,
this rational function cannot have roots in the open right-half plane, 
and it can neither have a root $\zeta=\rmi\omega\neq 0$ on the 
imaginary axis because
\[
   2\,\Real f(\rmi\omega) \,=\,  \beta m\,y^*\Chat_Y(\omega)y \qquad
   \text{for $\omega\in\R$}\,,
\]
cf.~\req{CYhat-kappa}, 
and $\Chat_Y(\omega)$ is assumed to be positive definite for $\omega\neq 0$.
This shows that $\Lambda$ is stable under the given assumption.
\fin
\end{remark}

\begin{remark}
\label{Rem:K}
\rm
A straightforward computation reveals that the inverse of the system
matrix $A$, when in the block form representation given in \req{Markov1}, 
takes the form
\bdm
   A^{-1} \,=\, 
   \begin{cmatrix}
      0 & 0 & I \\
      0 & \Lambda^{-1} & \Lambda^{-1}C \\
      -m\Omega_0^{-1} & m\Omega_0^{-1}B^T\Lambda^{-1} & m\Omega_0^{-1}B^T\Lambda^{-1}C
   \end{cmatrix}.
\edm
It follows that
\bdm
   A^{-1}E \,=\, \begin{cmatrix} 0 \\ 0 \\ -m\Omega_0^{-1} \end{cmatrix},
\edm
and hence,
\bdm
   E^TA^{-2}E \,=\, -m\Omega_0^{-1} \,.
\edm

On the other hand Proposition~\ref{Prop:Markov1} implies that
\bdm
   \int_0^\infty t\,C_Y(t)\dt 
   \,=\, \int_0^\infty t\,E^Te^{tA}\Sigma_0 E\dt
   \,=\, \frac{1}{\beta m}\,E^T\int_0^\infty t\,e^{tA}\dt \,E\,,
\edm
and partial integration thus gives
\bdmal
   \int_0^\infty t\,C_Y(t)\dt 
   &\,=\, \frac{1}{\beta m}\,E^T
          \left(tA^{-1}e^{tA}\Big|_{t=0}^\infty
                \,-\, \int_0^\infty A^{-1}e^{tA}\dt\right)E\\[1ex]
   &\,=\, -\frac{1}{\beta m}\,E^TA^{-2}e^{tA}\Big|_{t=0}^\infty E
    \,=\, \frac{1}{\beta m}\, E^TA^{-2}E\,.
\edmal
We thus conclude that
\[
   \int_0^\infty t\,C_Y(t)\dt \,=\, -\frac{1}{\beta}\,\Omega_0^{-1}\,.
\]
Accordingly, compare~\req{CVMoment1}, if $C_Y$ is a good enough 
approximation of the true autocorrelation function $C_V$ -- 
in the sense that the first moment of $C_Y$ (over the positive time axis) 
agrees well with the corresponding first moment of $C_V$ -- 
then $\Omega_0$ is a good approximation of $\Omega$.

Given the limited amount of data from which we start, this is not a
foregone conclusion. On the other hand, if $\Omega$ is known beforehand 
then it is easy to incorporate this information into the construction of the 
Ornstein-Uhlenbeck system~\req{Markov0}, namely by imposing
%, the same computation as above -- this time based on \req{CY} -- gives
%\be{CYMoment1b}
%   \int_0^\infty t\,C_Y(t)\dt 
%   \,=\, \frac{1}{\beta m}\,E^T\int_0^\infty t e^{tA}\dt E
%   \,=\, \frac{1}{\beta m}\,E^TA^{-2}E\,,
%\ee
%and hence, if we impose that
\be{EA2E}
   E^TA^{-2}E \,=\, -m\Omega^{-1}
\ee
as another constraint, so that $\Omega_0$ and $\Omega$ coincide.
\fin
\end{remark}

\begin{remark}
\label{Rem:fullprocess}
\rm
Take note that $C_{RV}$ is the antiderivative of $C_V$ with $C_{RV}(0)=0$
and $C_R$ is the second antiderivative of $-C_V$ with 
$C_R(0)=(\beta\Omega)^{-1}$. Likewise, $C_{XY}$ is the antiderivative of $C_Y$
with $C_{XY}(0)=0$ and $C_X$ is the second antiderivative of $-C_Y$
with $C_X(0)=(\beta\Omega_0)^{-1}$.
Accordingly, if we achieve a good match $C_Y\approx C_V$, then this will imply
a good match of $C_{XY}$ and $C_{RV}$ and their counterparts $C_{Y\!X}$ 
and $C_{V\!R}$. And this then implies that $C_X$ and $C_R$ are also close to
each other, provided that their initial values are in good agreement; 
otherwise we expect to see a systematic error in $C_R\approx C_X$ of
size $\norm{\Omega^{-1}-\Omega_0^{–1}}/\beta$.
\fin
\end{remark}

\section{Numerical construction of the Ornstein-Uhlenbeck process} 
\label{Sec:Implementation}
%%%%%%%%%%%%%%%%%%%%%%%%%%%%%%%%%%%%%%%%%%%%%%%%%%%%%%%%%%%
We now present the numerical scheme that we use to generate the
system~\req{Markov0} from given samples
\bdm 
   c_\nu \,=\, C_V(\nu\tau)\,, \qquad \nu=0,\dots,n\,,
\edm
for some $\tau>0$ and $n\in\N$ of the true autocorrelation function $C_V$.
As mentioned before, this algorithm is based on a method proposed 
in \cite{BBHSH26}; the differences are the particular side-constraints which
are being used here (i) to cope for the external quadratic potential, 
and (ii) to provide additional safeguards to increase the success rate of the 
overall algorithm.

According to \req{CY} the autocorrelation function of the $Y$-component
of the stationary solution of \req{Markov0} has the form
\be{acf-data}
   C_Y(t) \,=\, \frac{1}{\beta m}\,E^T e^{tA} E 
   \,=:\, \frac{1}{\beta m}\,\varphi(t)\,, \qquad t\geq 0\,.
\ee
When $A$ is diagonalizable with distinct eigenvalues
$\lambda_1,\dots,\lambda_p\in\C^-$, its eigendecomposition gives the 
representation 
\bdm
   A \,=\, \sum_{j=1}^p \lambda_j Q_j\,,
\edm
where $Q_j\in\C^{N\times N}$ are (oblique) projection matrices
onto the associated eigenspaces;
in our case, compare Section~\ref{Subsec:sysmat},
this assumption on $A$ is satisfied by construction, and each eigenspace 
has dimension $d$, generically.
%; some may degenerate to a smaller dimension, though. 
Using the above representation, $\varphi$ can be written as 
a finite Prony series
\be{Prony-series}
   \varphi(t) \,=\, \sum_{j=1}^p \Gamma_j e^{\lambda_j t}\,, \quad t\ge 0\,,
\ee
with coefficient matrices $\Gamma_j=E^TQ_jE\in\C^{d\times d}$;
compare, e.g., Higham~\cite{High08}.
%Take note here that $M$ has real entries, 
%and thus its eigenvalues are either real or come in complex conjugate pairs, 
%and the same applies to the associated matrices $Q_j$ and $\Gamma_j$. 

Since $\varphi(0)=I$ by virtue of \req{acf-data} we conclude
from \req{Prony-series} that
\be{EE-Gamma}
   \sum_{j=1}^p \Gamma_j \,=\, I\,.
\ee
In addition, the constraints \req{EM1E} and \req{EA2E}, which we want to impose
on $A$, translate into further constraints on the coefficient matrices 
$\Gamma_j$ in \req{Prony-series}, namely
\be{EM1E-Gamma}
   \sum_{j=1}^p \lambda_j\Gamma_j \,=\, 0 \qquad \text{and} \qquad
   \sum_{j=1}^p \lambda_j^{-1}\Gamma_j \,=\, 0\,,
\ee
respectively, and
\be{EA2E-Gamma}
   \sum_{j=1}^p \lambda_j^{-2}\Gamma_j \,=\, -m\Omega^{-1}\,.
\ee

In view of \req{acf-data} our goal now is to select a reasonable number 
$p\in\N$ of appropriate exponents $\lambda_j\in\C^{-}$ and 
associated coefficient matrices $\Gamma_j\in\C^{d\times d}$ such 
that $\varphi$ provides a good match of our given data samples, i.e.,
\be{PronyY}
    \varphi(\nu\tau) 
    \,=\, \sum_{j=1}^p \Gamma_j e^{\nu\tau\lambda_j} 
    \,\approx\, \beta m\,c_\nu \,=\, \beta m\,C_V(\nu\tau) 
\ee
for $\nu=0,\dots,n$, while satisfying the constraints \req{EE-Gamma}, 
\req{EM1E-Gamma} and -- subject to the availability of $\Omega$ -- 
\req{EA2E-Gamma}.

%%%%%%%%%%%%%%%%%%%%%%%%%%%%%%%%%%%%%%%
\subsection{Finding suitable exponents}
\label{Subsec:exps}
%%%%%%%%%%%%%%%%%%%%%%%%%%%%%%%%%%%%%%%
To determine appropriate exponents $\lambda_j$ we consider the generating
function
\be{gen-Func}
   F(z) \,=\, \sum_{\nu=0}^{\infty} c_\nu z^{-\nu-1}
\ee
of all equidistant snapshots $c_\nu$ of the underlying 
autocorrelation function $C_V$. This function is analytic in the exterior
of the unit circle, because $C_V$ is a bounded function, and hence
the Laurent coefficients $c_\nu$ of $F$ are also bounded. Let us assume 
for the moment that \req{PronyY} holds true with equality for all $\nu\in\N_0$. 
Then we can insert this in \req{gen-Func} to obtain
\[
\begin{aligned}
   F(z) &\,=\, \frac{1}{\beta m}\sum_{\nu=0}^{\infty} 
              \sum_{j=1}^p \Gamma_j e^{\nu\tau\lambda_j} z^{-\nu-1}
         \,=\, \frac{1}{\beta m}\sum_{j=1}^p \frac{\Gamma_j}{z} 
               \sum_{\nu=0}^\infty (e^{\tau\lambda_j}/z)^\nu\\[1ex]
        &\,=\, \frac{1}{\beta m}\sum_{j=1}^p \frac{\Gamma_j}{z - z_j}
\end{aligned}
\]
with
\bdm
z_j = e^{\tau\lambda_j}\,, \qquad j=1,\ldots,p \,.
\edm
Take note that this ansatz implies that $F$ is a real rational function
%i.e., $F(\overline{z})=\overline{F(z)}$ for all $z\in\C$, 
and that all the poles of $F$ are inside the unit circle.
We therefore proceed by approximating the generating function in the exterior 
of the unit disk by a rational function with the same two properties, using
the AAA algorithm by Nakatsukasa, S\`{e}te, and Trefethen~\cite{NST18} 
for this purpose;
more precisely, since $F$ is matrix-valued, we use a matrix-valued variant 
of this algorithm which we borrow from Gosea and Güttel~\cite{GG21}.

To be specific, let $\mathcal{Z}$ be the set of all points of an equiangular 
grid with an even number of grid points on a circle $|\zeta| = \rho > 1$ 
in the complex plane, with the two real grid points at $\pm\rho$ as 
anchor points. 
The AAA algorithm determines a rational function in barycentric form
\be{Bary}
   r(z) \,=\, \sum_{l=1}^p \frac{w_l F_l}{z - \zeta_l} \Big/ \sum_{l=1}^p 
   \frac{w_l}{z - \zeta_l} \,,
\ee
where $\zeta_l \in\mathcal{Z}$ are the so-called support 
points, and $w_l$ are suitable complex (scalar) weights. 
In the (generic) case that all weights are nonzero, the function \req{Bary} 
interpolates the given values $F_l\in\C^{d\times d}$ at all the support points, 
i.e.,
\bdm
   r(\zeta_l) \,=\, F_l\,, \qquad l=1,\ldots,p \,.
\edm
Therefore, in order to approximate $F$ of \req{gen-Func} we choose
\be{Fk}
   F_l \,=\, \frac{1}{\beta m}\sum_{\nu=0}^{n} c_\nu \zeta_l^{-\nu-1} 
   \,\approx\, F(\zeta_l) \qquad \text{for} \ \zeta_l\in\mathcal{Z}\,,
\ee
and we mention that the quality of this approximation depends on the size
of $\rho$ and the corresponding truncation error on the one hand, and on the 
magnitude of any noise in the data on the other; ideally the choice of $\rho$ 
is such that these two error components balance each other; 
see \cite{GrHa25} for a more detailed discussion of this matter.
Note that $F_l=\overline{F}_{l'}$ when $\zeta_l=\overline{\zeta}_{l'}$,
and hence, $r$ is real, if the support points are real or come in complex 
conjugate pairs, and if the associated weights have the same property;
we refer to \cite[Appendix~B.1]{BBHSH26} for a description how we determine
appropriate weights with this property.

We initialize the list of support points with the two real grid points
$\zeta_1 = \rho$ and $\zeta_2 = -\rho$, and then the AAA algorithm
proceeds with a greedy iterative scheme: in the $l$-th iteration, $l=1,2,\dots$,
two further grid points $\zeta_{2l+1}$ and $\zeta_{2l+2}=\overline\zeta_{2l+1}$ 
are added to the list of support points, 
where $\zeta_{2l+1}\in\mathcal{Z}$ is such that the Frobenius norm 
of the current residual $\norm{F_{2l+1} - r(\zeta_{2l+1})}_F$ is maximal,
%, and added to the list of support points; further. 
%the weights $w_k$ are chosen
%by minimizing a linearized fit 
%\bdm
%   \sum_{l=m+1}^M 
%  \left\Vert \sum_{k=1}^m w_k \frac{F_l - F_k}{\zeta_l - \zeta_k} \right\Vert_F
%\edm
%of the remaining points in $\mathcal{Z}$,
%subject to the constraint that the weight vector $w = [w_1,\ldots,w_m]^T$ has %unit Euclidean norm. The solutions of this minimization problem are
%singular vectors associated with the smalles singular value of the Loewner
%matrix\bdm
%L = \begin{cmatrix}
%        \frac{F^{m+1}_{1,1} - F^{1}_{1,1}}{\zeta_{m+1} - \zeta_1} & \dots & 
%        \frac{F^{m+1}_{1,1} - F^m_{1,1}}{\zeta_{m+1} - \zeta_m}\\
%        \vdots & & \vdots\\
%        \frac{F^M_{1,1} - F^1_{1,1}}{\zeta_{M} - \zeta_1} & \ldots & 
%        \frac{F^M_{1,1} - F^m_{1,1}}{\zeta_{M} - \zeta_m}\\[1ex]
%        \frac{F^{m+1}_{2,1} - F^1_{2,1}}{\zeta_{m+1} - \zeta_1} & \dots &
%        \frac{F^{m+1}_{2,1} - F^m_{2,1}}{\zeta_{m+1} - \zeta_m}\\
%        \vdots & & \vdots\\
%        \frac{F^M_{2,1} - F^1_{2,1}}{\zeta_{M} - \zeta_1} & \ldots & 
%        \frac{F^M_{2,1} - F^m_{2,1}}{\zeta_{M} - \zeta_m}\\
%        \vdots & & \vdots\\
%        \frac{F^{m+1}_{d,d} - F^1_{d,d}}{\zeta_{m+1} - \zeta_1} & \dots & 
%        \frac{F^{m+1}_{d,d} - F^m_{d,d}}{\zeta_{m+1} - \zeta_m}\\
%        \vdots & & \vdots\\
%        \frac{F^M_{d,d} - F^1_{d,d}}{\zeta_{M} - \zeta_1} & \ldots & 
%        \frac{F^M_{d,d} - F^m_{d,d}}{\zeta_{M} - \zeta_m}\\
%    \end{cmatrix} ,
%\edm
%where $F_{ij}^l$ stands for the $(i,j)$-matrix entry of $F_l$.
unless the residual norms $\norm{F_i - r(\zeta_i)}_F$
of all grid points $\zeta_i\in\mathcal{Z}$ are 
below a certain tolerance $\eps>0$ and a minimum number of admissible poles 
has been reached, in which case we stop the iteration.
Note that the minimum number of admissible poles depends on the number 
of side constraints, see Proposition~\ref{Prop:sdp} below. 

Once the iteration has been terminated we solve a generalized eigenproblem 
to compute the $p-1$ poles of the rational approximation~\req{Bary}, 
cf.~\cite{NST18}.
%which are given as the $m-1$ finite eigenvalues $z$ of the $(m+1)\times(m+1)$ 
%generalized eigenvalue problem
%\bdm
%\begin{cmatrix} 0 & w_1 & w_2 & \ldots & w_m\\ 1 & z_1 &&&\\ 1 && z_2 && \\ 
%\vdots &&& \ddots & \\ 1 &&&& z_m\end{cmatrix} = z \begin{cmatrix}
%        \,0 & \phantom{w_1} & \phantom{w_2} && \\ 
%        & 1 &&& \\ && 1 && \\ &&& \ddots & \\ &&&& 1\\
%    \end{cmatrix} \,.
%\edm
%Take note here that two further generalized eigenvalues of this problem are 
%infinite \textit{per se}. 
For every pole $z\neq 0$ with $|z|<1$ the exponent(s)
\be{logpole}
\lambda = \begin{cases}
        \displaystyle \frac{1}{\tau} \log z\,, & z\notin\R^{-}\,,\\[2ex]
        \displaystyle \frac{1}{\tau} \log |z| \pm \rmi\,\frac{\pi}{\tau}\,, & z\in\R^{-}\,, 
    \end{cases}
    \ \qquad |\Imag\lambda|\,\leq\, \frac{\pi}{\tau}\,,
\ee
are used for the Prony series ansatz~\req{Prony-series}, 
while spurious poles outside the unit disk and in the origin 
are ignored, because they are caused by
artefacts of the rational approximation. 
Since $r$ is real, the resulting exponents are either real or come in 
complex conjugate pairs. 

Take note that due to the definition \req{logpole} and due to spurious 
poles the value of $p$ in \req{Bary} may slightly differ from
the parameter $p$ in \req{Prony-series}, i.e., the number of terms of the
Prony series. To enhance readability we nevertheless took the liberty 
to use the same letter here and there.

%%%%%%%%%%%%%%%%%%%%%%%%%%%%%%%%%%%%%%%%%%%%%%%%%%
\subsection{Finding suitable coefficient matrices}
\label{Subsec:coeff}
%%%%%%%%%%%%%%%%%%%%%%%%%%%%%%%%%%%%%%%%%%%%%%%%%%
In a second step of our algorithm we determine coefficient matrices 
$\Gamma_j\in\C^{d\times d}$ for the Prony series $\varphi$ of 
\req{Prony-series} by minimizing the data fit
\[
   \sum_{\nu=0}^{n} \bigl\|\varphi(\tau\nu) - \beta m\,c_\nu \bigr\|_F^2
\]
in \req{PronyY}, subject to the aforementioned constraints \req{EE-Gamma} 
and \req{EM1E-Gamma}, together with \req{EA2E-Gamma} in case the stiffness 
matrix $\Omega$ is to be incorporated.

In view of our ultimate goal \req{acf-data} that $\varphi$ coincides 
-- up to the scalar factor $\beta m$ -- with the autocorrelation function 
of a stochastic process $Y$ on the positive time axis, 
we have to make sure that $\varphi$, when extended by
\be{phi-ext}
   \varphi(t) \,=\, \varphi(-t)^T\,, \qquad t<0\,,
\ee
to the negative time axis, is a function of positive type.
%respect Bochner's theorem which states that the Fourier transform 
%of every autocorrelation function is positive semidefinite for every 
%frequency $\omega\in\R$, cf., e.g., \cite{Pavl14}. 
%This requires that the Fourier transform of $\varphi$, given by
%\req{Prony-series} for $t\geq 0$, and extended by
%, must be positive semidefinite. 
A straightforward computation based on \req{Prony-series} shows that
\[
    \widehat{\varphi}(\omega) 
    \,=\, \sum_{j=1}^p 
          \Bigl(\frac{\Gamma_j}{\rmi\omega - \lambda_j}   
                \,-\, \frac{\Gamma_j^T}{\rmi\omega+\lambda_j}
          \Bigr)\,, \qquad \omega\in\R\,,
\]
and hence, this expression has to be positive semidefinite for every 
$\omega\in\R$. In order to make this side condition numerically tractable we 
limit our attention to the (a priori known) zero eigenvalues of 
$\phihat$ at $\omega=0$ and $\omega=\infty$. For high frequencies 
$\phihat(\omega)$ can be expanded into the convergent Laurent series
\bdm
   \widehat{\varphi}(\omega) \,=\, \sum_{k=0}^{\infty} \YG_k \omega^{-k-1}
\edm
with selfadjoint coefficient matrices
\bdm
   \YG_k \,=\, (-\rmi)^{k+1}\sum_{j=1}^p 
          \lambda_j^k\Bigl(\Gamma_j + (-1)^{k+1} \Gamma_j^T\Bigr)\,,
%   &\,=\, (-\rmi)^{n+1}
%          \Bigl(\varphi^{(n)}(0+) \,+\, (-1)^{n+1}\varphi^{(n)}(0+)^T\Bigr)\,,
    \qquad k\in\N_0\,.
\edm
Accordingly, a necessary condition for $\phihat(\omega)$ to be positive 
semidefinite for large frequencies is that
the first non-vanishing term of this Laurent series is
positive semidefinite. Note that $\YG_0=\YG_1=0$ according to \req{EE-Gamma} 
and \req{EM1E-Gamma}, respectively. As a consequence, if $\YG_2$ happens to 
have a nonzero eigenvalue, then $\phihat(\omega)$ exhibits 
a negative eigenvalue when $\omega$ is large and has the appropriate sign. 
%\bdm
%   \YG_2/\rmi \,=\, \varphi''(0+) \,-\, \varphi''(0+)^T
%\edm
%is a skew-symmetric real matrix; therefore $\YG_2$ is semidefinite, 
%if and only if it is the zero matrix, i.e., if and only if $\varphi''(0+)$ 
%is symmetric. Making use of \req{Prony-series} again, we 
We therefore add the requirements that
\[
   \YG_2/\rmi \,=\, \sum_{j=1}^p \lambda_j^2(\Gamma_j - \Gamma_j^T) \,=\, 0
   \qquad \text{and} \qquad
   \YG_3 \,=\, \sum_{j=1}^p \lambda_j^3(\Gamma_j + \Gamma_j^T)
         \,\succcurlyeq\, 0 
\]
to our list of constraints.

Likewise, near the origin $\phihat(\omega)$ can be expanded into the
Taylor series 
\bdm
   \widehat{\varphi}(\omega) \,=\, \sum_{k=0}^{\infty} \PsiG_k \omega^{k}
\edm
with 
\bdm
   \PsiG_k \,=\, -\rmi^k\sum_{j=1}^p 
                 \lambda_j^{-k-1}\Bigl(\Gamma_j + (-1)^k\Gamma_j^T\Bigr)\,,
   \qquad k\in\N_0\,.
\edm
Here, $\PsiG_0=\PsiG_1=0$ because of \req{EM1E-Gamma} and the fact that
\bdm
   \sum_{j=1}^p \Gamma_j\lambda_j^{-2} 
   \,=\, \beta m\int_0^\infty t\,C_Y(t)\dt \,=\, -m \Omega_0^{-1}\,,
\edm
compare Remark~\ref{Rem:K}, with $\Omega_0$ being symmetric 
according to Proposition~\ref{Prop:Markov1}. We therefore impose 
the additional constraint that
\[
   \PsiG_2 \,=\, \sum_{j=1}^p \lambda_j^{-3}\bigl(\Gamma_j + \Gamma_j^T)
   \,\succcurlyeq\, 0\,. 
\]

In summary, the optimization problem we consider takes the following form:
%\be{sdp-problem}
%\begin{aligned}
%    &\mbox{Minimize} && 
%    \sum_{\nu=0}^{2n-1} \left\Vert 
%       \sum_{j=1}^m \Gamma_j e^{\lambda_j\tau\nu} - Y_\nu \right\Vert_F\\
%    &\mbox{subject to} && \sum_{j=1}^m \Gamma_j \,=\, \frac{1}{\beta m}\,I\,, 
%    && \sum_{j=1}^m \lambda_j\Gamma_j \,=\, 0\,,
%    && \sum_{j=1}^m \frac{1}{\lambda_j}\,\Gamma_j \,=\, 0\,,\\
%    & && \sum_{j=1}^m \lambda_j^2(\Gamma_j - \Gamma_j^T) \,=\, 0\,, 
%    && \sum_{j=1}^m \lambda_j^3(\Gamma_j + \Gamma_j^T)\,\succcurlyeq\, 0 \,, \\
%    &\mbox{and} && 
%   \sum_{j=1}^m \frac{1}{\lambda_j^2}\,\Gamma_j\,=\,-\frac{1}{\beta}\,\Omega^{-1}\,,
%    && \mbox{when available}\,.
%\end{aligned}
%\ee
Find coefficient matrices $\Gamma_j\in\C^{d\times d}$, $j=1,\dots,p$,
which minimize
\begin{subequations}
\label{eq:sdp-problem}
\be{sdp-objective}
    \sum_{\nu=0}^{n} 
       \Biggl\Vert 
          \sum_{j=1}^p \Gamma_j e^{\lambda_j\tau\nu} - \beta m\,c_\nu
       \Biggr\Vert_F^2,
\ee
subject to the equality constraints
\be{sdp-constraints}
\begin{aligned}
    \sum_{j=1}^p \Gamma_j &\,=\, I\,, 
    & \quad \sum_{j=1}^p \lambda_j\Gamma_j &\,=\, 0\,,\\
    \sum_{j=1}^p \frac{1}{\lambda_j}\,\Gamma_j \,&=\, 0\,,
    & \quad \sum_{j=1}^p \lambda_j^2(\Gamma_j - \Gamma_j^T) &\,=\, 0\,, 
\end{aligned}
\ee
the semidefiniteness constraints
\be{sdp-semidefconstraints}
   \sum_{j=1}^p \lambda_j^3(\Gamma_j + \Gamma_j^T)\,\succcurlyeq\, 0 \,, 
   \qquad
   \sum_{j=1}^p \lambda_j^{-3}\bigl(\Gamma_j + \Gamma_j^T)
   \,\succcurlyeq\, 0\,,
\ee
and the additional equality constraint
\be{sdp-optional} 
    \phantom{xxxxx}
    \sum_{j=1}^p \lambda_j^{-2}\Gamma_j\,=\,-\frac{1}{\beta}\,\Omega^{-1}\,,
    \qquad \mbox{when $\Omega$ is prescribed}\,.
\ee
\end{subequations}

\begin{proposition}
\label{Prop:sdp}
Assume that the number $p$ of terms of the 
Prony series~\req{Prony-series} satisfies $7\leq p \leq n+1$.
Then the optimization problem~\req{sdp-problem} has a solution 
$\{\Gamma_j\}_{j=1}^p\subset\C^{d\times d}$, and
this solution gives rise to a real-valued Prony series~\req{Prony-series}.
This solution is unique if 
%at most one negative pole is present in 
the definition~\req{logpole} of the exponents of the Prony series
involves not more than one negative pole of the rational function.
If there is more than one negative pole in \req{logpole} then there
may be infinitely many real-valued Prony series~\req{Prony-series} 
whose coefficient matrices solve the optimization problem~\req{sdp-problem}; 
any two such solutions differ by a term of the form
\be{Prop:sdp}
   \psi(t)\sin(\pi t/\tau)\,, \qquad t\geq 0\,,
\ee
which vanishes at all grid points, where snapshots of $C_V$ have been taken.
The factor $\psi$ in \req{Prop:sdp} is a linear combination of 
exponentials $e^{\mu_k t}$ with negative parameters $\mu_k$ and coefficient
matrices in $\R^{d\times d}$.
\end{proposition}

\begin{proof}
The constraints \req{sdp-constraints}, \req{sdp-semidefconstraints}, 
and \req{sdp-optional} define a closed convex set~$\mathcal{C}$.
To see that $\mathcal{C}$ is nonempty, consider the subset 
$\mathcal{C}_0\subset\mathcal{C}$, which is obtained, when replacing the
last constraint in \req{sdp-constraints} and the two semidefiniteness
constraints~\req{sdp-semidefconstraints} by the more restrictive 
equality constraints
\bdm
   \sum_{j=1}^p \lambda_j^2\Gamma_j \,=\, 0\,, \qquad
   \sum_{j=1}^p \lambda_j^3\Gamma_j \,=\, 0\,,
   \qquad \text{and} \qquad
   \sum_{j=1}^p \lambda_j^{-3}\Gamma_j \,=\, 0\,.
\edm
The resulting seven constraints are inhomogeneous linear equality contraints of 
Vandermonde type, which have a solution, i.e., $\mathcal{C}_0\neq\emptyset$,
if the rectangular transposed Vandermonde matrix
\bdm
   \begin{cmatrix}
      1 & \lambda_1 & \cdots & \lambda_1^6\\
      1 & \lambda_2 & \cdots & \lambda_2^6\\
      \vdots & \vdots & & \vdots\\
      1 & \lambda_p & \cdots & \lambda_p^6
   \end{cmatrix} \,\in\,\C^{p\times 7}
\edm
has a trivial null space, and this is the case whenever $p\geq 7$.

The objective function~\req{sdp-objective} is convex on $\mathcal{C}$.
Therefore the minimization problem~\req{sdp-problem} always has a 
solution if $p\geq 7$.
Let $\{\Gamma_j\}$ be such a solution, and define another set of coefficient 
matrices $\{\widetilde\Gamma_j\}$ via
\bdm
   \widetilde\Gamma_j \,=\, \overline\Gamma_{j'}\,, \quad
   \text{where \ $\lambda_{j'}=\overline\lambda_{j}$}\,; \qquad j=1,\dots,p\,.
\edm
Then it is easy to see that $\{\widetilde\Gamma_j\}$ also belongs to 
$\mathcal{C}$, and that the objective function attains the same minimal value 
for this second set of coefficient matrices. Due to convexity the same is
true for the coefficient matrices $\{(\Gamma_j+\widetilde\Gamma_j)/2\}$,
for which the associated Prony series is real-valued.

Two different real-valued Prony series, whose coefficient 
matrices solve \req{sdp-problem}, must differ by a function which vanishes 
at all grid points $\nu\tau$, $\nu=0,\dots,n$, 
Since $p\leq n+1$ such a 
situation can only occur when some of the exponents of the two Prony series 
have an imaginary part $\pm\pi/\tau$; compare, for example, 
\cite{PPST18} in combination with \cite[Remark~5.2]{Hank24}.
According to \req{logpole} this (only)
happens when the associated poles of the rational function are negative. 
This shows that if there are two Prony series $\varphi_1$ and $\varphi_2$ of 
this sort then $\varphi_2-\varphi_1$ is a function of the form~\req{Prop:sdp}, 
where $\psi$ is a linear combination of exponentials $e^{\mu_kt}$
and the exponents $\mu_k$ are the real parts of the corresponding exponents 
$\lambda_k$ of the Prony series. From this it follows that
\bdm
   \varphi_2'(0) \,-\, \varphi_1'(0) \,=\, \frac{\pi}{\tau}\,\psi(0)\,,
\edm
and hence, $\psi(0)=0$ due to the second constraint in \req{sdp-constraints}.
If there is only one negative pole then 
\bdm
   \psi(t) \,=\, \psi(0) e^{\mu t}\,, \qquad t\geq 0\,,
\edm
for some fixed value of $\mu$, and therefore \req{Prop:sdp} vanishes
identically, showing that the minimizer of \req{sdp-problem} is unique in
this case.
\end{proof}

\begin{remark}
\label{Rem:p}
\rm
The lower bound $p\geq 7$ is a conservative estimate of the necessary
number of terms to guarantee feasibility of the list of constraints.
For example, if $d=1$ then the equality constraints in
\req{sdp-constraints} amount to three scalar equations in total, 
because the last one is always true; 
together with \req{sdp-optional} this gives four scalar equations, 
which can all be satisfied as soon as $p\geq 4$.
For higher dimensions $p=5$ terms of the Prony series are always sufficient
to match all the equality constraints in \req{sdp-constraints} and 
\req{sdp-optional}, and it is conceivable that
%. With a little bit of luck \textcolor{red}{klingt ein bisschen salop}, 
the semidefiniteness
constraints in \req{sdp-semidefconstraints} are also satisfied for the 
corresponding parameters.
\fin
\end{remark}

In our code we solve \req{sdp-problem} by first ignoring the 
semidefiniteness constraints \req{sdp-semidefconstraints},
and solving the resulting linearly constrained least-squares problem by 
classical linear algebra techniques, cf., e.g., Bj\"orck~\cite{Bjor24}. 
If the corresponding solution happens to satisfy \req{sdp-semidefconstraints} 
then this is the solution of the full problem~\req{sdp-problem}. 
%Otherwise we use a common technique to rewrite the 
%minimization problem~\req{sdp-problem} as a semidefinite program,
%cf., e.g., Vandenberghe and Boyd~\cite{VB96}, 
%and use the software package {\sc cvx} \cite{CVX, GB08} for its solution; 
%the latter implements interior point methods for solving convex optimization 
%problems.
Otherwise we add the constraints from \req{sdp-semidefconstraints}
which are violated by this approximation,
and use a common technique to rewrite the resulting minimization problem as a 
semidefinite program, cf., e.g., Vandenberghe and Boyd~\cite{VB96}. Then we 
use the software package {\sc cvx} \cite{CVX,GB08} 
%for its solution \textcolor{red}{vielleicht eher to solve the  problem oder 
to obtain a solution; {\sc cvx} employs
%the latter implements 
interior point methods for solving convex optimization problems. If 
the corresponding solution does still not satisfy \req{sdp-semidefconstraints}, 
we then solve the full problem~\req{sdp-problem} accordingly.

%%%%%%%%%%%%%%%%%%%%%%%%%%%%%%%%%%%%%%%%%%%%%%%%
\subsection{Construction of the Langevin system}
\label{Subsec:sysmat}
%%%%%%%%%%%%%%%%%%%%%%%%%%%%%%%%%%%%%%%%%%%%%%%%
In order to achieve the desired representation \req{acf-data} for a
particular function $\varphi$ given by \req{Prony-series} we have to 
determine coefficient matrices $A$ and $L_0$ of a suitable 
Langevin system~\req{Markov0}. 

For the construction of the (stable diagonalizable)
drift matrix $A$ we refer to the algorithm described in detail 
in \cite[Appendix~B.2]{BBHSH26}. The dimension of $A$ is given by
\bdm
   N \,=\, \sum_{j=1}^p \operatorname{rank}(\Gamma_j) \,\leq\, d\cdot p\,,
\edm
where the inequality arises because
%the inequality being due to the fact \textcolor{red}{klingt unschön} that 
singular subspaces associated 
with any null space components of the coefficient matrices $\Gamma_j$ are 
discarded. This is necessary to ensure that the realization $(A,E,E)$ 
of the transfer function $\kappa$ of \req{kappa} is minimal 
-- using the language of system theory, cf.~\cite[Section~3.4]{AV73}.
In this case the Positive Real Lemma~\cite[Chapter~5]{AV73} 
states that  if and only if the transfer function is 
positive real, the (singular) Lur'e equations
\be{singLure}
   AP \,+\, PA^T \,=\, -HH^T\,, \qquad PE \,=\, E\,,
\ee
have real solution matrices $P$ and $H$, where $P\in\R^{N\times N}$ is
symmetric positive definite.
On the other hand the transfer function is positive real, 
if and only if the extension~\req{phi-ext} of $\varphi$ 
is a function of positive type; compare~\req{acf-data} and \req{CYhat-kappa}. 
In other words, the Lur'e equations have a solution,
if and only if the semidefiniteness constraints~\req{sdp-semidefconstraints},
which we have imposed on the coefficients of $\varphi$,
have been sufficient to determine an admissible
Prony series approximation \req{Prony-series} of the given data.
Moreover, cf.~\cite[Section~5.5]{AV73}, if this happens to be the case
then there is a solution of \req{singLure} with $H\in\R^{N\times d}$,
and due to the zero upper left block of $A$ this matrix must have the form
\bdm
   H \,=\, \sqrt{\beta m}\,\begin{cmatrix} 0 \\ L_0 \end{cmatrix}
\edm
for some $L_0\in\R^{(N-d)\times d}$. 
When this particular $L_0$ is used in \req{Markov0},
then $P/(\beta m)$ is the covariance matrix of the stationary 
solution of \req{Markov0}, and the autocorrelation function $C_Y$ of the 
$Y$ component of this solution coincides with $\varphi/\beta m$ as desired.

As shown in Section~\ref{Sec:Markov1}, a suitable variable transformation
turns the matrix $A$ into the equivalent form given in \req{Markov},
%matrix of the equivalent system~\req{Markov}, which then takes the form
%\be{Markov2}
%   \rmd\begin{cmatrix} Y\\ Z\\ X \end{cmatrix} 
%   \,=\, \begin{cmatrix}
%            \,0 & \,B^T & -\Omega_0/m \\ -C\, & \Lambda & \phantom{i}0 \\ 
%            \,I & 0 & \phantom{i}0
%         \end{cmatrix}
%         \begin{cmatrix} Y\\ Z\\ X \end{cmatrix}\!\dt
%         \,+\, \begin{cmatrix} 0 \\ L \\ 0 \end{cmatrix}\!\dW\,,
%\ee
and the stationary solution of \req{Markov} provides approximations of 
the velocity and the position of the macroparticle.
Note that, as shown in Proposition~\ref{Prop:Markov1},
the Lur'e equations~\req{singLure} are equivalent to the reduced 
Lur'e equations~\req{Lure-small}. Rather than \req{singLure} we therefore
solve the latter ones for the remaining entry $L$ of the 
fluctuating force term in \req{Markov}, together with
the covariance matrix $\Sigma_{22}$ of the auxiliary variables $Z$.

The numerical solution of the singular Lur'e equations is rather involved 
and has been treated elsewhere. We refer to 
Wang, Speyer, and Weiss~\cite{WSW90}, for example; 
see also \cite[Appendix~C]{BBHSH26}.
%for the generic case that $\varphi''(0+)$ 
%and $\YG_3$ of \req{M2M3} are positive definite.

%%%%%%%%%%%%%%%%%%%%%%%%%%%%%%%%%%%%%%%%%%%%%%%%
\subsection{Position autocorrelation data}
\label{Subsec:PACF}
%%%%%%%%%%%%%%%%%%%%%%%%%%%%%%%%%%%%%%%%%%%%%%%%
Finally, we briefly comment on the case that only position autocorrelation
data 
\bdm
   C_R(\nu\tau)\,, \qquad \nu=0,\dots,n\,,
\edm
are given, in which case $\Omega=(\beta C_R(0))^{-1}$ is always known; 
compare Remark~\ref{Rem:MLL16}. 
As follows from the discussion in Remark~\ref{Rem:fullprocess},
the autocorrelation functions $C_Y$ and $C_X$ of our searched-for 
Ornstein-Uhlenbeck process~\req{Markov} satisfy
\be{CXpp}
   C_X''(t) \,=\, -C_Y(t)\,.
\ee
Accordingly, if $C_Y$ is given by \req{acf-data}, \req{Prony-series}, 
with the constraints \req{EE-Gamma}, \req{EM1E-Gamma}, and \req{EA2E-Gamma}
being satisfied, then
\be{CX-Prony}
   C_X(t) \,=\, -\frac{1}{\beta m}\, 
                \sum_{j=1}^p \lambda_j^{-2}\Gamma_j e^{\lambda_jt}\,, \qquad
   t \geq 0\,.
\ee
We therefore replace the generating function~\req{gen-Func} by
\bdm
   F(z) \,=\, \sum_{\nu=0}^\infty C_R(\nu\tau)\,z^{-\nu-1}\,,
\edm
use the AAA algorithm as described in Section~\ref{Subsec:exps}
to determine a rational approximation of this function on a circle 
$|z|=\rho>1$, and take its poles to define the exponentials of the Prony
series~\req{Prony-series}; cf.~\req{logpole}.

In view of \req{CX-Prony} we then determine the coefficients $\Gamma_j$
in \req{Prony-series} as in Section~\ref{Subsec:coeff}, replacing the
objective function~\req{sdp-objective} by
\bdm
    \sum_{\nu=0}^{n} 
       \Biggl\Vert 
          \sum_{j=1}^p \lambda_j^{-2}\Gamma_j e^{\lambda_j\tau\nu} 
          + \beta m\,C_R(\tau\nu)
       \Biggr\Vert_F^2.
\edm
Beware of the fact that the two terms in the norm are now coupled with a plus
sign because of the minus sign in \req{CX-Prony}.
Further note, that while the stiffness matrix $\Omega$ is readily available
in this case, the mass matrix is not and therefore has to be provided as 
additional information.

A potential advantage of using position data is that in some applications
position data decay more slowly than velocity data and therefore provide
a better signal to noise ratio; on the other hand,
high frequency information in the velocity data is damped in the position
autocorrelation function by virtue of \req{CXpp}.

%%%%%%%%%%%%%%%%%%%%%%%%%%%%%%%%%%%%%%%%%%%%%%%%%%%%%%%%%%%
\section{Numerical results} 
\label{Sec:Numerics}
%%%%%%%%%%%%%%%%%%%%%%%%%%%%%%%%%%%%%%%%%%%%%%%%%%%%%%%%%%
\subsection*{Example 5.1}
As our first example we consider the translation and rotation invariant system
of a spherical particle in a harmonic trap, which gives rise to one-dimensional 
processes $V$ and $R$. The detailed setup of this system and the generation
of corresponding molecular dynamics data for the four cross-correlation
functions $C_V$, $C_R$, $C_{VR}$, and $C_{RV}$ is specified in 
Appendix~\ref{App}. These four scalar functions
are shown as black dashed lines in Figure~\ref{fig:1D}. 
%\textcolor{red}{In this example the (scalar) stiffness parameter is 
%$\Omega=100$.}

\begin{figure}
    \centering
    \includegraphics[width=\linewidth]{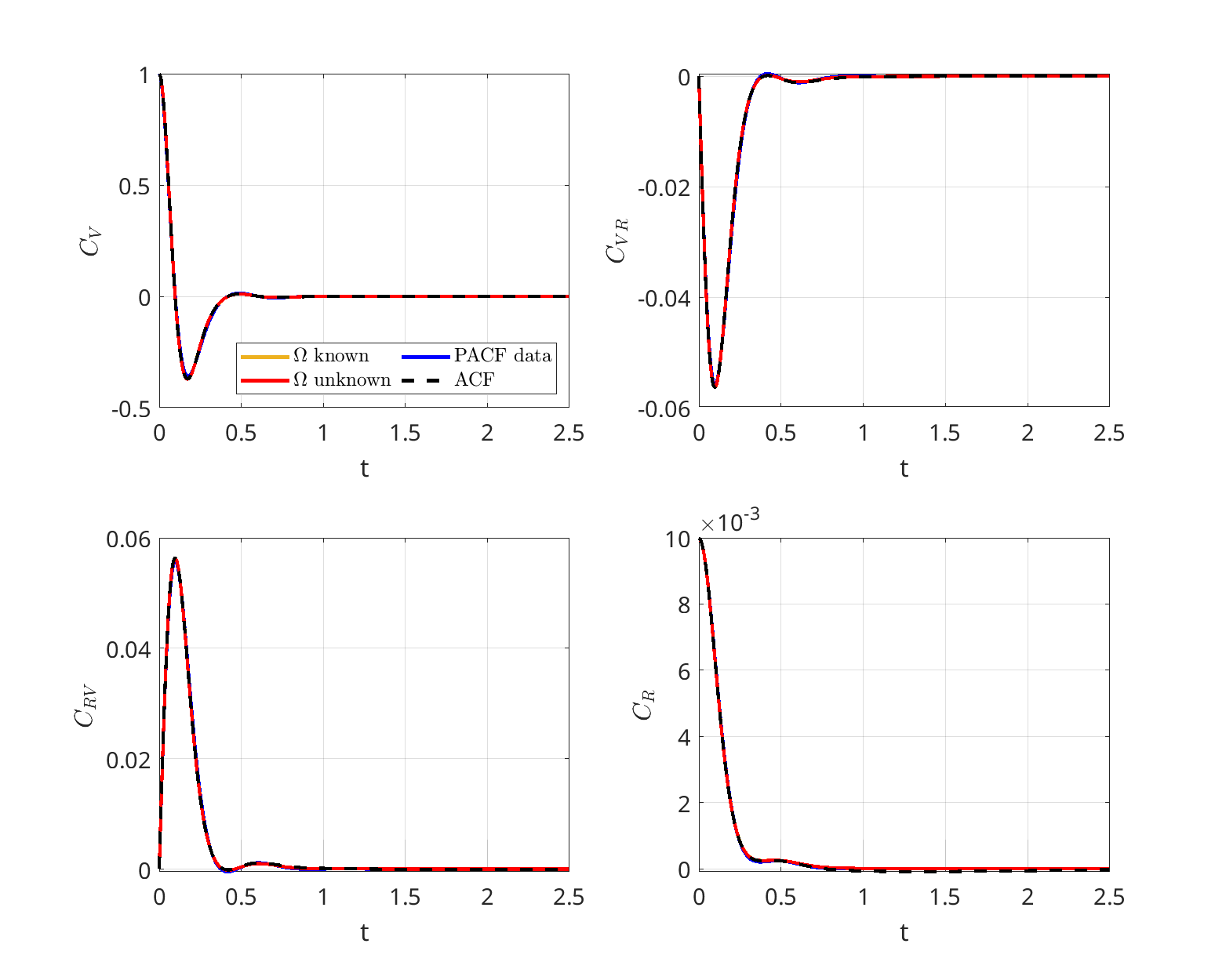}
    \caption{Autocorrelation function of Example 5.1 versus time $t$ and 
    the three numerical approximations. Each panel shows the scalar graphs
    corresponding to the respective entry of the matrix-valued function.}
    \label{fig:1D}
\end{figure}

We now discuss constructions $X\approx R$ and $Y\approx V$ computed by the 
algorithms proposed above. For our first two algorithms we use 
$n+1=41$ samples 
of the velocity autocorrelation function with a grid spacing 
$\tau = 0.025$ as input; the corresponding grid therefore covers
the initial time interval $[0,1]$. These first two algorithms differ in that
the first one assumes that $\Omega$ is known a priori, so we enforce 
the constraint~\req{sdp-optional} for this approximation, while
the second one assumes that $\Omega$ is unknown. Accordingly, the resulting
two approximations $C_Y$ of $C_V$ use the same number of terms in the 
Prony series~\req{Prony-series} and the same exponents, 
because the latter are determined in the first phase of our algorithm, 
which is the same for both methods. In this first phase, where we employ
the AAA algorithm, cf.~Section~\ref{Subsec:exps},
we evaluate the generating function on an equiangular grid with 100 grid points 
on a circle with radius $\rho = 1.15$. 
After initializing the list of support points with the two real grid points 
at $\pm\rho$, the AAA algorithm performs three iterations in this example
and determines a rational function with seven admissible poles to reduce the
residuals below the tolerance $\eps=10^{-4}$.
One of these poles is positive, 
and the remaining ones come in three complex conjugate pairs. 
The two resulting Prony series approximations of the velocity
autocorrelation function thus have $p=7$ terms, each.

We determine the (scalar) coefficients $\Gamma_j$ of these Prony series
as described in Section \ref{Subsec:coeff}. Recall from 
Proposition~\ref{Prop:sdp} that the associated constrained 
least-squares problems~\req{sdp-problem} have unique solutions, each,
and note that according to Remark~\ref{Rem:p} we have seven free parameters
to satisfy three or four scalar equations in the present one-dimensional case,
depending on whether \req{sdp-optional} is included as a constraint or not.
The remaining degrees of freedom provide the flexibility to optimize the 
least-squares fit.
In either case, i.e., with and without the constraint~\req{sdp-optional}, 
the solution of the linearly constrained linear least squares problems
is no feasible solution of the full convex problem~\req{sdp-problem}, 
because $\YG_3$ happens to be negative. Therefore the semidefinite program 
is invoked to solve the optimization problem;
since $\PsiG_2$ has been positive in this example 
we only add the constraint $\YG_3\geq 0$ to the previous equality constraints.
In both cases the software {\sc cvx} returns
an approximate solution of \req{sdp-problem} which is feasible
in that both, $\YG_3$ and $\PsiG_2$, are positive.
Also, the singular Lur'e equations~\req{singLure} turn out to be solvable 
in both cases, so the two Prony series approximations yield valid 
autocorrelation functions.
The resulting two Ornstein-Uhlenbeck systems have dimension $N=d\cdot p=7$, each, 
and therefore they come with $7-2=5$ auxiliary variables.

The third algorithm uses samples of the position autocorrelation function
instead of velocity data; compare Section~\ref{Subsec:PACF}.
These data are sampled on a grid with the same grid spacing $\tau=0.025$ as
before, but this time we only use $n+1 = 31$ data points corresponding to 
the time interval $[0,0.75]$. The reason is -- as can be seen from 
Figure~\ref{fig:1D-error} (see the text below) -- that the autocorrelation 
function reaches the noise level near $t=0.75$ while the velocity 
autocorrelation function stays above the noise level for all $t\in [0,1]$. 
Accordingly, we also choose a larger value $\rho=1.2$ for computing
interpolation data~\req{Fk} for the generating function to achieve a similar 
(relative) truncation level as before despite the reduced number of 
available terms in \req{Fk}.
Here we run the AAA algorithm with the smaller stopping tolerance 
$\eps=10^{-6}$ to cope for the fact that $C_R(0)=(\beta\Omega)^{-1}=1/100$
is two orders of magnitude smaller than the sample values of the velocity 
autocorrelation function that have been used above. 
The AAA algorithm performs two  
iterations and returns a rational function with five admissible poles,
one of them being positive and the remaining four being complex conjugate 
pairs. So the resulting Prony series approximation of $C_V$ has $p = 5$ terms.
It turns out that the solution of the corresponding linearly constrained 
linear least squares problem is feasible in that $\YG_3$ and $\PsiG_2$ 
are both positive. The solution of the corresponding Lur'e equations yields
an Ornstein-Uhlenbeck process of dimension $N = d\cdot p = 5$, which means that the
corresponding system has three auxiliary variables.

\begin{figure}
    \centering
    \includegraphics[width=\linewidth]{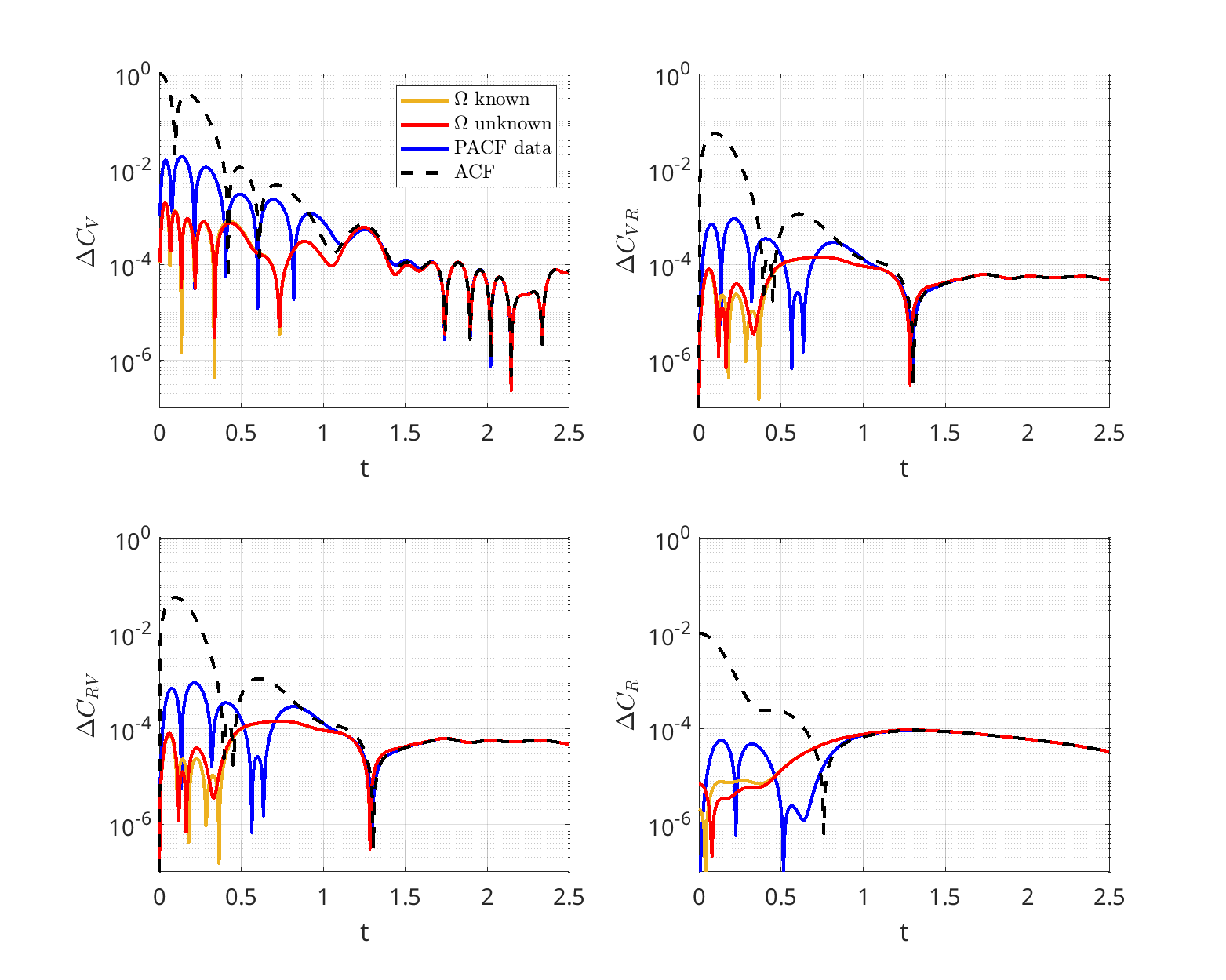}
    \caption{Absolute error between the true autocorrelation function of 
    Example 5.1 and the three approximations. Each panel contains the scalar graphs corresponding to the error of the 
    respective entry of the matrix-valued approximation. 
    The black dashed lines show the
    target functions for additional information.}
    \label{fig:1D-error}
\end{figure}

The associated autocorrelation functions are depicted in 
Figure~\ref{fig:1D}: The orange graphs correspond to the functions
where velocity autocorrelation data and the value of $\Omega$ are being used, 
the red graphs correspond to velocity autocorrelation data without
the use of $\Omega$, and the blue ones are the approximations based on
position autocorrelation data. 
%Each panel corresponds to the respective scalar entries of the $2\times 2$ 
%matrix-valued function. 
In these plots the blue and orange lines are almost fully covered by 
the red ones, which in turn show a very good alignement with the original
data (in black). For a better evaluation of the quality of these approximations 
we refer to Figure \ref{fig:1D-error}, which shows 
the absolute differences between the respective approximation and 
the true values in a semilogarithmic scale; also included here are the 
reference functions (in black); otherwise, the color coding 
is the same as in Figure~\ref{fig:1D}. Returning to what we have indicated
before, the black lines decay down to an absolute noise level of 
about $10^{-4}$; for the velocity autocorrelation function (top left panel) 
this is reached near $t=1.5$, whereas the position autocorrelation function 
(bottom right panel) is down to the noise level already at about $t=0.75$. 
Since the Prony series decay exponentially their errors are dominated by the
reference function values, as soon as the latter have reached the noise level.
%As is already suggested from the insets in Figure~\ref{fig:1D} the errors
%from $t=1$ onwards are dominated by the given data of the functions, 
%because the Prony series approximations vanish exponentially fast, while 
%the given data have slower asymptotic decay 
%(cf., e.g., \cite[Section~9.1]{Zwan01}) and, above all, are dominated by 
%statistical errors. 
One can further see that the two approximations from velocity data 
are significantly better than the third one in blue by almost one 
order of magnitude. This reflects the situation that more samples of the
velocity data could be used, which led to more terms in the corresponding 
Prony series and thus to a better data fit; compare the final remark in
Section~\ref{Subsec:PACF}.
%\textcolor{red}{Another reason for the inferior results when using position 
%data may be that the position autocorrelation function is the second
%antiderivative of the negative velocity autocorrelation function, 
%cf.~Remark~\ref{Rem:fullprocess}, and hence, 
%information in the latter is being smoothed out in the position data.}
The two approximations based on velocity autocorrelations in turn have fairly 
similar quality; only in an initial time interval the match of the position 
autocorrelation data differs because of the slight misfit of the 
$\Omega$-approximation: the corresponding estimate 
$C_X(0)=(\beta\Omega_0)^{-1}$ can be seen to differ from the true value 
$C_R(0)=0.01$ by a relative error of $10^{-3}$, roughly, 
which is similar to the errors in the approximations of $C_V$. 
But aside of that, in this example the incorporation of the additional 
constraint~\req{sdp-optional} does not lead to a significant benefit.

As in \cite{BBHSH26} we have varied the number of grid points to test
the robustness of our algorithm. Note that our algorithm is not fail-proof
in that we only incorporate a necessary condition for our approximation
$\varphi$ to be of positive type; compare Section~\ref{Subsec:coeff}. 
If $\varphi$ fails to be of positive type then the Lur'e equations 
have no solution, and we cannot generate a corresponding Ornstein-Uhlenbeck process. 
It turns out that for this example, the robustness of the
method is very similar to the one in \cite{BBHSH26}: our algorithm has been 
successful for the vast majority of tests that we have made.

\subsection*{Example 5.2}
\begin{figure}
    \centering
    \includegraphics[width=1\linewidth]{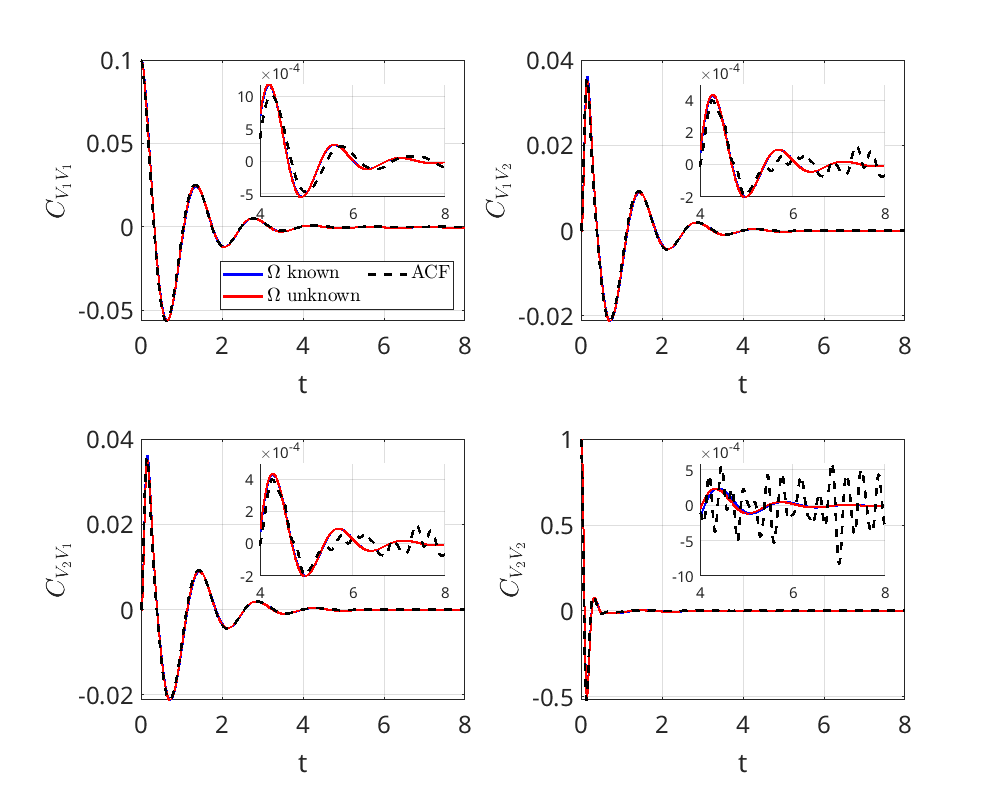}
    \caption{Velocity autocorrelation function of Example 5.2 versus
    time and the two approximations with and without the integral constraint 
    \req{sdp-optional}. Each panel shows the scalar graphs corresponding to the 
    respective entry of the matrix-valued function.}
    \label{fig:2D-VACF-neu}
\end{figure}
Next we consider the dynamics of two coupled particles of different 
masses in a harmonic trap, which gives rise to two-dimensional processes $V$ 
and $R$, each. Again we refer to Appendix~\ref{App} for a detailed description 
of the physical system. In this example we focus on the use of velocity data, 
which are displayed in Figure~\ref{fig:2D-VACF-neu},
because the numerical results for position data have been inferior in 
Example~5.1. More precisely, we use $n+1 = 31$ samples of $C_V\in\R^{2\times 2}$ 
with a grid spacing of $\tau = 0.1$ as input data, covering the initial time 
interval $[0,3]$; as can be seen from the insets in 
Figure~\ref{fig:2D-VACF-neu}, near $t=4$ the individual entries of $C_V$ 
reach the absolute noise level of about $5\cdot 10^{-4}$.

This time we need to preprocess the given data, 
because the mass matrix is not a multiple of the identity matrix, 
compare \req{mass}. To this end we rescale the velocity autocorrelation 
function $C_V$ by computing the Cholesky decomposition $C_V(0) = L_CL_C^T$ and 
multiplying $C_V(t)$ by $L_C^{-1}$ and $L_C^{-T}$ from left and right, 
respectively.
%The rescaled velocity autocorrelation function $C_V$ of this process is shown 
%as black dashed line in Figure \ref{fig:2D-VACF}. \textcolor{red}{No, don't!}
After preprocessing the data, the corresponding generating function is
evaluated for the AAA scheme on an equiangular grid with $100$ grid points 
on a circle, this time with radius $\rho = 1.5$.
The AAA algorithm performs four iterations to reduce the residuals below the 
error tolerance of $\varepsilon = 5\cdot10^{-4}$, 
and provides a matrix-valued rational function with seven admissible poles. 
One of these poles is negative and is thus duplicated according to 
\req{logpole}; the other six poles consist of two real and two complex 
conjugate pairs. The Prony series ansatz thus has $p=8$ terms in this 
particular example.

Take note here that since we only have one negative pole and $p\ge 7$, 
the optimization problem~\req{sdp-problem} for the coefficient matrices 
of the Prony series has a unique minimizer according to 
Proposition~\ref{Prop:sdp}.
Again we compute two approximations, one with and one without the integral 
constraint \req{sdp-optional}. In both cases the solution of the linearly
constrained linear least squares problem is not in the feasible set of the
overall convex optimization problem, because $\YG_3$ fails to be positive 
definite; We add the violated constraint to the 
list of constraints and use {\sc cvx} for the solution of the augmented 
optimization problems. This leads to a feasible approximate minimizer in
the case of using \req{sdp-optional}, but fails in the second case. 
The reason is that now $\YG_3$ is positive definite as requested, 
but $\PsiG_2$ is not. We therefore restart {\sc cvx} for this instance 
with both semidefiniteness constraints
and finally obtain feasible approximate minimizers of \req{sdp-problem}
for both cases. The corresponding singular Lur'e equations \req{singLure} 
are solvable, and therefore the two Prony series approximations provide valid 
autocorrelation functions $C_Y\approx C_V$.
The resulting Ornstein-Uhlenbeck system has dimension $N=d\cdot p = 16$, which corresponds
to $16 - 4 = 12$ auxiliary variables. 

\begin{figure}
    \centering
    \includegraphics[width=\linewidth]{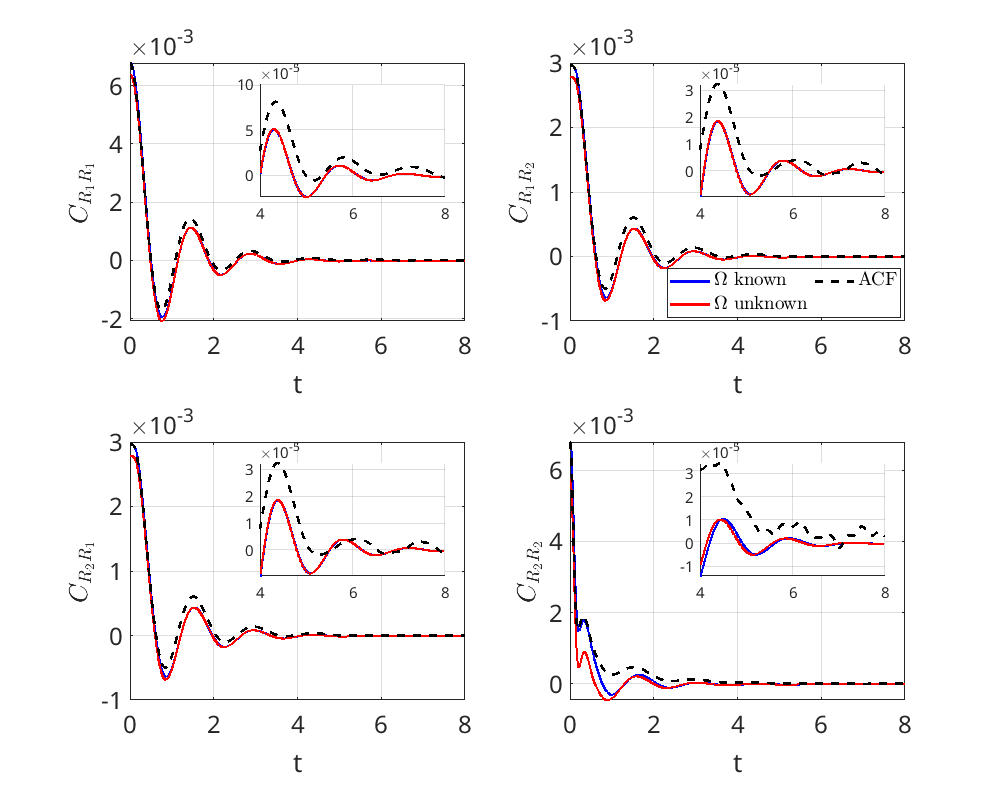}
    \caption{Position autocorrelation function of Example 5.2 versus time and 
    the two approximations with and without the 
    integral constraint~\req{sdp-optional}. Each panel shows the scalar graphs 
    corresponding to the respective entries of the matrix-valued function.}
    \label{fig:2D-PACF-neu}
\end{figure}

See Figure~\ref{fig:2D-VACF-neu} for a plot of $C_V$ and the two approximations
$C_Y$ and Figure~\ref{fig:2D-PACF-neu} for the corresponding 
position autocorrelation functions.
%we omit the cross-correlations $C_{RV}$ and $C_{VR}$
Here the blue lines correspond to the functions where $\Omega$ is known 
and the red lines correspond to the other ones, respectively, and the 
reference functions are black and dashed. 
It can be seen that the velocity autocorrelation functions provide an excellent
match for both algorithms, while the approximations of the 
position autocorrelation functions are slightly inferior. 
We omit plots of the cross-correlations; the fits are of comparable quality.

\begin{figure}
    \centering
    \includegraphics[width=\linewidth]{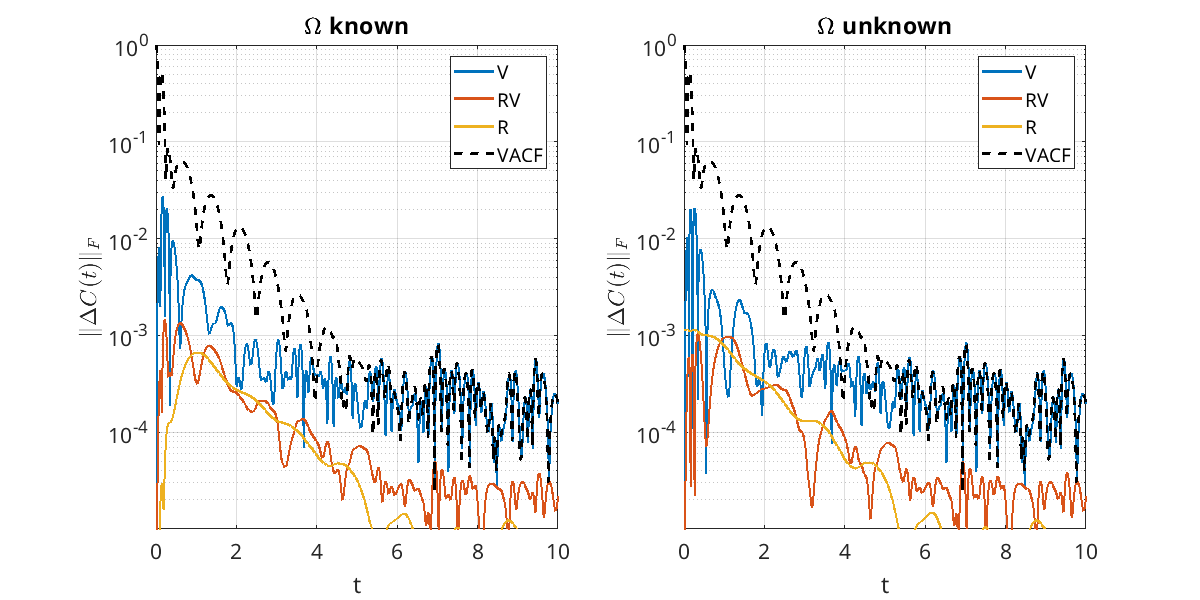}
    \caption{Absolute error between the given velocity, position and cross correlation of Example 5.2 and the two approximations. The left panel shows the approximation error where $\Omega$ is known, the right panel where $\Omega$ is unknown, respectively.}
    \label{fig:2D-Error-neu}
\end{figure}

Concerning the use of the $\Omega$-constraint~\req{sdp-optional} one can
only see a benefit in the lower right panel of Figure~\ref{fig:2D-PACF-neu};
otherwise the two approximations are of very similar quality.
This is also confirmed in Figure~\ref{fig:2D-Error-neu}, which provides
a better quantitative picture of the error of the respective approximations.
In this figure we display the Frobenius norm of the errors
in the $2\times 2$ correlation functions $C_V$, $C_R$, and $C_{RV}$
versus time. Here the left-hand panel corresponds to the approximation where
$\Omega$ is prescribed, whereas the right-hand panel corresponds to the 
approximation which presumes $\Omega$ to be unknown.
We observe similar error histories in both panels; only for $t<1$ the error
in $C_R$ is better in the left panel, where $\Omega$ has been prescribed.
Overall, the absolute error in $C_R$ is smaller by about one order of
magnitude, but the individual entries of $C_R$ in itself are also smaller 
than those of $C_V$ by (at least) the same amount.
% approximation
%Furthermore the method is quite stable for different number of grid points $n$.

\section{Conclusion}
In this work we have presented a framework for constructing Markovian 
embeddings for generalized (multidimensional) Langevin equations 
with an external harmonic potential from autocorrelation data. 
We have further discussed conditions for stationary solutions of these
equations. The algorithm which we suggest is an extension of a method
presented in \cite{BBHSH26} to cope for the harmonic potential in the
system. This method determines a rational approximation with a matrix-valued 
version of the AAA algorithm of the generating function of the given data
and uses the poles of this rational function to choose the exponential terms
of a Prony series approximation of the data. The coefficients of this 
expansion are determined by solving a linear least-squares problem with 
semidefiniteness constraints.
Numerical experiments with one- and two-dimensional molecular dynamics data
have been shown to be very successful and indicate that the method has a 
high potential to generalize faithfully to more complex physical data.

\appendix
\section{Details of the molecular dynamics simulation}
\label{App}
The test data for the two examples in Section~\ref{Sec:Numerics} have been
generated with the software package LAMMPS~\cite{thompson2022}
(stable release April 2, 2025). 
To be specific we have simulated a liquid of $1000$ particles interacting 
with a Weeks-Chandler-Andersen potential~\cite{andersen1991} 
using the TVN-ensemble,
where the mass of the particles, the thermal energy, and the characteristic 
energy and length scale parameters of the interaction potential are unity 
in the reduced unit system.
The cubic simulation box with periodic boundary conditions has an edge
length of $10.7722$ spatial units corresponding to a particle density of $0.8$.
The initial configuration, with all particles on a regular 
Cartesian grid with unit grid spacing, has
%Initial configurations have 
been set up using the molecular builder Moltemplate~\cite{jewett2021}.

For the one dimensional system discussed in Example~5.1 one of the particles 
(our particle of interest) is attached to its starting position by 
a harmonic spring with spring constant $\Omega=100$.
To remove the linear momentum introduced into the liquid by this external
force the center of mass motion of the liquid is set to zero every 
$1000$~steps. Using a timestep of $0.005$ we integrate both systems for 
$5\cdot 10^4$ steps with a
Langevin thermostat (cf., e.g., D\"unweg and Paul~\cite{dunweg1991}) 
using a thermostat constant of one time unit to obtain 
equilibrated structures. Then, to obtain 50 independent configurations of 
each system we randomize the velocity of all particles and integrate for 
another $10^5$ steps with a Nosé-Hoover-type thermostat
(Tuckerman et al~\cite{tuckerman2006}) using a thermostat constant of 
$5.0$ time units.
Finally we integrate these 50 independent systems for $10^7$ steps, each, 
saving the position and velocity of the trapped particles at every time step. 
The correlation functions are averaged over all $50$ systems 
and all three spatial dimensions.

For the two dimensional system we use the same setting, except that 
we set the time step to $0.0025$ and attach two neighboring particles to 
their starting positions with harmonic springs with spring constant 
$\Omega=100$, each, and on top of that, connect these two particles 
with an elastic rod with spring constant $\Omega=50$.
We modify the mass of the second particle to be $m=10$, i.e., ten times heavier
than the first one, and we switch off the Weeks-Chandler-Andersen interaction
between the two particles. Finally, we constrain two coordinates of 
the two particles to their initial values and only allow their motion in the
third coordinate axis of the Cartesian grid; accordingly the velocities
of the two tagged particles are scalar functions, each.
Take note that in this example so-called depletion layers around the 
two particles under consideration give rise to additional external forces
acting on these particles, 
cf., e.g., Lekkerkerker, Tuinier, and Vis~\cite{LTV24},
and therefore the quadratic form~\req{Phi} of the
potential $\PhiG$ is just its leading order approximation.

%The effective potential between the two particles is determined both by the 
%added springs and by the effective interaction between two otherwise non-
%interacting particles in a liquid, which is commonly employed in free-energy 
%methods such as umbrella sampling.\cite{KaeUS2011}

%%%%%%%%%%%%%%%%%%%%%%%%%%%%%%%%%%%%%%%%%%%%%%%%%%%%%%%%%%%

%%%%%%%%%%%%%%%%%%%%%%%%%%%%%%%%%%%%%%%%%%%%%%%%%%%%%%%%%%%
\end{document}